\documentclass[a4paper, 12 pt, twoside, reqno]{amsart}

\usepackage{slashed}
\usepackage[english]{babel}
\usepackage[utf8]{inputenc}

\usepackage[all]{xy}
\usepackage{setspace}
\usepackage[euler-digits,euler-hat-accent]{eulervm}

\usepackage{times}
\usepackage{amsmath}
\usepackage{amsfonts}
\usepackage{amssymb}
\usepackage{amsmath}
\usepackage{amsthm}
\usepackage{graphicx}
\usepackage{array}
\usepackage{color}
\usepackage{mathrsfs}
\usepackage{hyperref}
\usepackage{eucal}
\usepackage{esint} 
\usepackage{tikz}
\usepackage{upgreek}
\usepackage{enumitem}
\usepackage{mathtools}

\allowdisplaybreaks

\theoremstyle{plain}
\newtheorem{theorem}{Theorem}[section]
\newtheorem{corollary}[theorem]{Corollary}
\newtheorem{proposition}[theorem]{Proposition}
\newtheorem{lemma}[theorem]{Lemma}

\theoremstyle{definition}
\newtheorem{definition}[theorem]{Definition}

\theoremstyle{remark}
\newtheorem{remark}[theorem]{Remark} 
\newtheorem{example}[theorem]{Example}

\numberwithin{equation}{section}
\numberwithin{figure}{section}
\numberwithin{table}{section}

\newcommand{\h}{\mathfrak{h}}
\newcommand{\R}{\mathbb{R}}
\newcommand{\N}{\mathbb{N}}
\newcommand{\C}{\mathbb{C}} 

\newcommand{\Z}{\mathbb{Z}}
\newcommand{\T}{\mathbb{T}}
\newcommand{\rz}{\mathtt{R}}
\newcommand{\U}{\mathtt{U}}

\newcommand{\s}[1]{\CMcal{#1}}
                  
\newcommand{\bb}[1]{\mathscr{#1}}

\newcommand{\n}[1]{\mathbb{#1}}

\newcommand{\fb}{\mathfrak{f}_B}

\newcommand{\K}{\s{K}}

\newcommand{\expo}[1]{\,\mathrm{e}^{#1}\,}

\newcommand{ \ii}{\,\mathrm{i}\,}

\newcommand{\virg}[1]{\lq\lq#1\rq\rq}                \newcommand{\ie}{\textsl{i.\,e.\,}}

\newcommand{\cf}{\textsl{cf}.\,}

\newcommand{\A}{\mathfrak{A}}

\hypersetup{
pdftoolbar=true,        
pdfmenubar=true,        
pdffitwindow=true,     
pdfstartview=true,    
pdftitle={NCG-Landau-II},    
pdfauthor={Giuseppe De Nittis},     
breaklinks=true, 
colorlinks=true,       
linkcolor=purple,         
citecolor=teal, 
urlcolor=blue, 
bookmarksopen=true, 
filecolor=magenta,      
}

\begin{document}

\title[Interface currents and corner states in magnetic quarter-plane systems]{Interface currents and corner states in magnetic quarter--plane systems
}

\author[D. Polo]{Danilo Polo Ojito}

\address[D. Polo]{Facultad de Matem\'aticas,
  Pontificia Universidad Cat\'olica de Chile,
  Santiago, Chile.}
\email{djpolo@mat.uc.cl}

\vspace{2mm}

\date{\today}

\begin{abstract}
We study the propagation of currents along the interface of two $2$-$d$ magnetic systems, where one of them occupies the first quadrant of the plane.  By considering the tight-binding approximation model and K-theory, we prove that, for an integer number that is given by the difference of two bulk topological invariants of each individual system, such interface currents are quantized.  We further state the necessary conditions to produce corner states for these kinds of underlying systems, and we show that they have topologically protected asymptotic invariants.

\medskip


\end{abstract}

\maketitle

\tableofcontents

\section{Introduction}\label{sect:intro}

A \emph{magnetic quarter-plane system} is composed of two $2$-$d$ materials subjected to constant magnetic fields so that one material occupies the first quadrant of the plane and the other material the remaining part of the plane; see fig. \ref{fig:my_label} and \ref{eq:def_B}. When the magnetic fields are orthogonal and of distinct intensity, there may exist extended states near the interface between the two materials carrying currents\footnote{Similar to the edge currents in the quantum Hall effect}. Furthermore, under suitable time-dependent perturbation of the system could appear localized states at the corner of the interface. This work deals with the study of topological invariants associated with \emph{interface currents} and \emph{corner states} in magnetic quarter-plane systems for which the interactions between particles are neglected. These new results are a physical manifestation of the K-theory and index theory of certain $C^*$-algebras.

\begin{figure}[ht]
    \centering
    \includegraphics[scale=0.5]{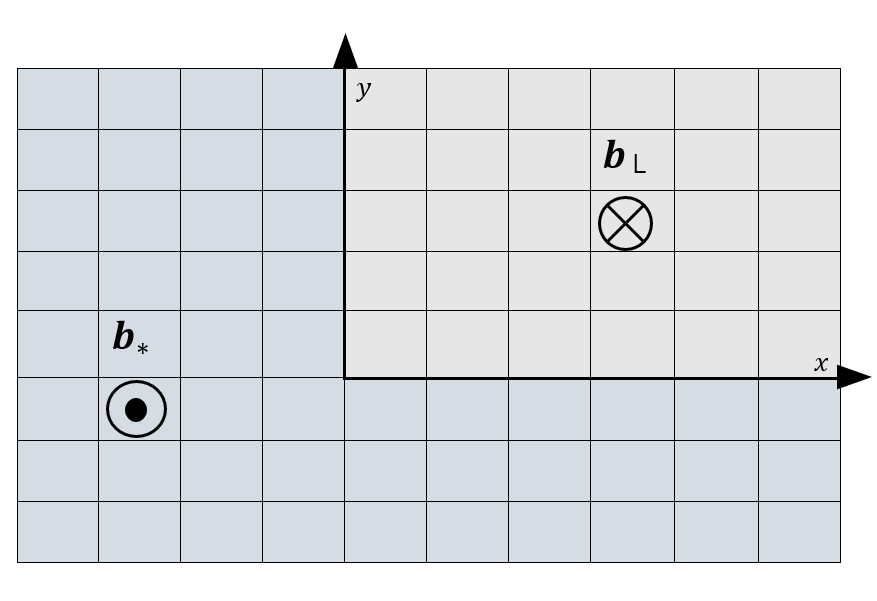}
    \caption{A picture of the magnetic quarter-plane. The interface is the thin region that separates the two materials. The magnetic effects in the full system are modeled by a magnetic field oriented orthogonally to the plane with strength function $B\colon\Z^2\rightarrow \R$ such that has the value $b_\llcorner$ in the first quadrant and $b_\ast$ otherwise.}
    \label{fig:my_label}
\end{figure}
\subsection{Interface currents} There are prior rigorous works about interface currents by considering the \emph{Iwatsuka magnetic system} \cite{Iwa}, \ie each material occupies a half-plane and the Iwatsuka magnetic field is a $y$-independent magnetic field oriented orthogonally to the plane with a strength function $B_I\colon \R^2\rightarrow\R$ so that $\lim\limits_{x\rightarrow\pm\infty}B_I(x)=b_\pm$. An interesting property is that in the Iwatsuka case the interface currents are quantized by the difference of two \emph{Chern numbers} of the two systems. A rigorous mathematical justification for the latter fact  was obtained in \cite{Dom} for continuous systems and in \cite{Kot} for discrete systems.  Using K-theory applied to the Iwatsuka Toeplitz extension, De Nittis and Gutierrez \cite{Deni1} also obtained the same result for discrete systems. It is important to point out that the work \cite{Deni1} also set a mathematical framework for analyzing general magnetic systems with interfaces.

In the tight-binding approximation, we prove that there may exist two topological currents flowing along the faces of the interface in the magnetic quarter-plane systems and that such currents are quantized by an integer number given by the difference of two bulk invariants of each material (Theorem \ref{Teo current}). 
 Following the previous work \cite{Deni1}, the main ingredient to get the latter result is to construct a suitable \emph{quarter-plane exact sequence} (sequence \ref{seq a}) which relates the states localized near the interface with the bulk states of the two subsystems. Then the interface and bulk conductances\footnote{The conductance associated to the interface and bulk currents of the system. We interpret the bulk conductance as the two bulk topological invariants of each material since it provides the value of the Hall conductance in the bulk of a quantum Hall system \cite{Bel1,Bel,Kun,Tho}.} are encoded in elements of some K-groups associated to the quarter-plane exact sequence. Consequently, we show the bulk-interface correspondence of these conductances based on the ideas of \cite{Guo}.  It is worth noting that the proofs also cover the case of the classical quarter-plane obtained in \cite{Guo}  when one of the two materials is chosen to be the vacuum.
Furthermore, we also prove that this result is robust against localized geometrical imperfection of the materials and even disorder effects can be handled, however, we prefer not to discuss this issue.

\subsection{Corner states} One important result in the study of the \emph{quarter-plane algebras} and its K-theory is the following exact sequence
\begin{equation}\label{seq to}
     \xymatrix{
 0\ar[r]&\mathcal{K}\ar[r]^{} & \mathcal{T}^{\alpha,\beta} \ar[r]^{}& \mathcal{S}^{\alpha,\beta}\ar[r]&0}\;,
\end{equation}
where $\mathcal{T}^{\alpha, \beta}$ is the quarter-plane Toeplitz algebra, $\mathcal{K}$ is the algebra of compact operators, and $\mathcal{S}^{\alpha,\beta}$ is the pullback $C^*$-algebra of two suitable Toeplitz algebras. The latter sequence was proved by Douglas and Howe \cite{Dou} in a restricted context, and after that, Park obtained the general case \cite{Par1, Par2}.  In 2018, Hayashi \cite{Hay} defines two topological invariants associated with the sequence \ref{seq to}: one is defined by counting corner states of the system and its non-triviality implies that the corner Hamiltonian\footnote{Periodic family of selfadjoint elements in $\mathcal{T}^{\alpha,\beta}$} is gapless. The other invariant is defined in terms of the spectral gap of the bulk Hamiltonian\footnote{Periodic family of selfadjoint elements in $\mathcal{S}^{\alpha,\beta}$}. Both invariants are elements of some K-groups associated with the sequence \ref{seq to} and there is a topological correspondence between these two invariants \cite[Theorem 3.7]{Hay}.

In the same vein as Hayashi, we construct a Toeplitz extension for the magnetic quarter-plane algebra (Theorem \ref{Teo seq}), and we define two topological invariants associated with such a sequence: the corner invariant that describes the existence of states localized near the point $(0,0)$, and the gapless invariant related with the asymptotic behavior of the system.  Afterward, using the six-term exact sequence on the sequence \ref{Teo seq}, we show that these invariants are in correspondence (Theorem \ref{teo: corner}), and a necessary condition for to obtain non-trivial corner states is that the time-dependent magnetic Hamiltonian is gapless (Corollary \ref{coro fin}).

\medskip

\noindent
{\bf Structure of the paper.}
Section 2 is devoted to the construction of the magnetic quarter-plane algebra and its magnetic hull. In more detail, we give an explicit presentation of the magnetic hull which allows us to describe the magnetic quarter-plane algebra in terms of crossed products. In section 3 we calculate the K-groups of the magnetic quarter-plane algebra, the bulk algebra, and the interface algebra. Moreover, a description of the interface algebra in terms of suitable projections is provided. In section 4, we present the bulk-interface correspondence. We begin this section by introducing the bulk and interface conductances. Afterward, we state the necessary assumptions needed to obtain the bulk-edge correspondence. This section finishes with the proof of the quantization of the interface currents. Section 5 is devoted to the study of the corner states. Indeed, we construct  a suitable Toeplitz extension for the magnetic quarter-plane algebra and after that, we give a characterization of the Fredholm operators in this $C^*$-algebra. Finally, we define two topological invariants associated with the quarter-plane Toeplitz extension and we prove that there is a correspondence between these invariants.

\medskip

\noindent
{\bf Acknowledgements}
I would like to cordially thank G. De Nittis for suggesting me this problem, and for his guidance and support throughout this work. I also would like to express my gratitude to N. J. Buitrago for reading the preliminary version of this paper and for helpful remarks. I have special thanks to E. Gutierrez, J. Gomez, M. Moscolari and S. Teufel for many stimulating discussions. In particular, I am indebted to M. Moscolari for his help in building explicit examples. This research was supported by ANID-Subdirección de Capital Humano/ Doctorado Nacional/ 2022-21220144.  

\section{Magnetic quarter-plane algebra}
This section is devoted to the detailed study of \emph{magnetic quarter-plane algebra} and its \emph{magnetic hull}. For an introduction to the construction of  magnetic $C^*$-algebras, we refer to \cite{Deni1}. 
\subsection{Magnetic fields and Vector Potentials}
In the tight-binding approximation for  two-dimensional magnetic systems, the position space is $\mathbb{Z}^2$ and an orthogonal magnetic field is nothing more than a function $B\colon\mathbb{Z}^2 \rightarrow \mathbb{R}$. For our case of interest let us consider the \emph{quarter-plane  magnetic field} given by
\begin{equation}\label{eq:def_B}
    B(n)\;:=\; \left\{ \begin{array}{lll}
             b_\ast &\quad &  \text{if} \quad n_1 \leq 0\,\,\text{or}\,\,n_2\leq 0\;, \\
           b_{\llcorner} &\quad&  \text{otherwise}\,, 
             \end{array}
   \right.
\end{equation}
where $b_{\llcorner},b_\ast\in\R$ so that $b_{\llcorner}-b_\ast\notin 2\pi\Z$\footnote{
This condition is necessary to get non-trivial topological effects. Otherwise, the flux operator (see \ref{eq:f_B}) is constant.  } and $n=(n_1,n_2)\in \Z^2.$ For every magnetic field defined in $\Z^2$ one can find infinite vector potentials $A_B\colon\Z^2\times\Z^2\rightarrow\R$ whose circulation is exactly $B$, namely such that $B(n) =\mathtt{Cir}[A_B](n)$, for all $n\in\Z^2$, where
\begin{equation*}
    \begin{split}
     \mathtt{Cir}[A_B](n)\;:=\; &A_B(n,n-e_1)+A_B(n-e_1,n-e_1-e_2)\\&+A_B(n-e_1-e_2,n-e_2)+A_B(n-e_2,n)  \;.
    \end{split}
\end{equation*}
Here $e_1:=(1,0)$ and $e_2:=(0,1)$ denote the canonical basis of $\Z^2$.
A \emph{standard} choice for the vector potential associated with the magnetic field $B$ is given by
   \begin{equation}\label{eq:A_B}
       A_B(n,n-e_j)\;:=\;\delta_{j,1}\left(\delta_{n_2>0}\sum_{m=1}^{n_2}B(n_1,m)-\delta_{n_2<0}\sum_{m=0}^{|n_2|-1}B(n_1,-m)\right)\;,
       \end{equation}
        where $n=(n_1,n_2)\in \Z^2$. The check that the circulation of this vector potential provides $B$ is a matter of a straightforward calculation.
 \begin{remark}
 Since the magnetic field \eqref{eq:def_B} is not invariant under translation in some direction, it follows that there are not \emph{Landau-type vector potentials} (\cf \cite[Section 2.2]{Deni3}) for the magnetic field $B$.
 \end{remark}  
 We can introduce the magnetic translations $\mathfrak{s}_1$ and $\mathfrak{s}_2$ associated to the given magnetic potential $A_B$ as
  \begin{equation}\label{eq:mag_tras}
    (\mathfrak{s}_j\psi)(n)\;:=\;\expo{\ii A_B(n,n-e_j)  }\psi(n-e_j)\;,\qquad j=1,2
    \end{equation} 
where $\psi\in \ell^2(\Z^2)$. Let $f_B\colon\Z^2\to \n{U}(1)$ be the function
\begin{equation}\label{eq:f_B}
f_B(n)\;:=\;\expo{\ii B(n)}\;,\qquad n\in\Z^2
\end{equation}
which provides the  exponential of  the \emph{magnetic flux} through the  unit cell sited in $n$, and define the associated \emph{flux operator} as
\begin{equation}
    (\mathfrak{f}_B\psi)(n)\;:=\; f_B(n)\psi(n)\;,\qquad \psi\in \ell^2(\Z^2)\;.
\end{equation}
A straightforward computation shows that the magnetic translations
 satisfy the condition
\begin{equation}\label{equ 2.4}
    \mathfrak{s}_1\;\mathfrak{s}_2\;\mathfrak{s}_1^*\;\mathfrak{s}_2^*\;=\;\mathfrak{f}_B\;.
\end{equation}

\medskip

Once the magnetic translations have been defined, let us introduce the magnetic algebra associated with them.
\begin{definition}
The \emph{magnetic quarter-plane} algebra $\A_B$ (in the gauge $A_B$ fixed by \eqref{eq:A_B}) is the  $C^*$-algebra inside the bounded operators $\mathcal{B}(\ell^2(\Z^2))$ generated by the magnetic translations $\mathfrak{s}_j$ defined by \eqref{eq:mag_tras}, \ie
$$\A_B\;:=\;C^*(\mathfrak{s}_1,\mathfrak{s}_2)\;.$$
\end{definition}
From its definition, it turns out that $\A_B$ is an unital $C^*$-algebra.
\begin{remark}
Since there are infinite magnetic potentials whose circulation is exactly $B$, one could be worried about the loss of generality in the definition of the magnetic algebra $\A_B$.  This is however not the case, since every pair of magnetic potentials, let us say $A_B$ and $A_B'$, with equal circulation, are connected by a gauge function $G_B:\Z^2\rightarrow \R$ according to
$$
A_B(n,m)\; =\;A_B'(n,m) +G_B(n)-G_B(m)\;,\qquad|n-m|\;=\; 1\;.
$$
Therefore, if the magnetic translations associated with the magnetic potentials $A_B$ and $A_B'$ are $\mathfrak{s}_{j}$ and $\mathfrak{s}'_{j}$ respectively, then it holds that 
$$
\mathfrak{s}'_{j}\;=\; \expo{-\ii G}\;\mathfrak{s}_{j}\;\expo{\ii G}\;,\qquad j= 1,2\;.
$$
As a consequence, the associated magnetic algebras are  unitarily equivalent. 
\end{remark}

\subsection{Magnetic Hull}
Let $\mathfrak{m}\colon\ell^\infty(\Z^2)\rightarrow \mathcal{B}(\ell^2(\Z^2))$ be the injective $\ast$-morphism which associates to each $g\in\ell^\infty(\Z^2)$ the multiplication operator $\mathfrak{m}_g$ defined by
$$
(\mathfrak{m}_g\psi)(n)\;:=\;g(n)\psi(n)\;,\qquad \psi\in \ell^2(\Z^2)\;.
$$
An example of this type of operator is given by the flux operator $\mathfrak{f}_B=\mathfrak{m}_{f_B}$.
 Consider the $\Z^2$-action defined over $\ell^\infty(\Z^2)$ by
\begin{equation}
    \tau_\gamma(g)(n)\,:=g(n-\gamma)\;,\qquad n,\gamma\in \Z^2\;,
\end{equation}
for every $g\in \ell^\infty(\Z^2)$. It holds true that
\begin{equation}\label{equ 2.6}
\tau_\gamma\left(\mathfrak{m}_{g}\right)\;:=\;    \mathfrak{m}_{\tau_\gamma(g)}\;=\;(\mathfrak{s}_1)^{\gamma_1}\;(\mathfrak{s}_2)^{\gamma_2}\;\mathfrak{m}_g\;(\mathfrak{s}_2)^{-\gamma_2}\;(\mathfrak{s}_1)^{-\gamma_1}
\end{equation}
 with $\gamma=(\gamma_1,\gamma_2)$.
In view of  \eqref{equ 2.4} and \eqref{equ 2.6} one gets that 
$$
\tau_\gamma(\mathfrak{f}_B)\;=\;\mathfrak{m}_{\tau_\gamma({f}_B)}\;\in\; \A_B\;,\qquad \forall\;\gamma\in \Z^2\;.
$$ 
Therefore, the $C^*$-algebra
\begin{equation}
    \mathfrak{F}_B\;:=\;C^*\big(\mathfrak{1},\{\tau_\gamma(\mathfrak{f}_B)\;|\;\gamma\in \Z^2\}\big)\;\subset\; \A_B
\end{equation}
 is a sub $C^*$-algebra of $\A_B$ that is commutative and with unit. Moreover,
by construction, it is invariant  under the $\Z^2$-action implemented by the translations $\tau$.
  It follows from Gelfand-Naímark isomorphism that there exists a compact Hausdorff space $\Omega_B$ such that $\mathfrak{F}_B\simeq C(\Omega_B)$. Such space $\Omega_B$ will be called the \emph{magnetic hull} of the algebra $\A_B$.  In the following, we will focus on the description of the space $\Omega_B$.
 
\medskip

Given a subset $\mathtt{O}\subseteq \Z^2$ let $\chi_\mathtt{O}$ be the \emph{characteristic function} of $\mathtt{O}$ and $\mathfrak{m}_{\chi_\mathtt{O}}$ the associate multiplication operator. Evidently $\mathfrak{m}_{\chi_\mathtt{O}}$ is a projection. 
\begin{proposition}\label{Prop projections}
Given $n\in \Z^2$ define the sets 
\begin{equation}
    \begin{split}
       & \mathtt{R}_n\;:=\;\{n+je_1\,|\,j\in \N_0\}\;,\\
        & \mathtt{U}_n\;:=\;\{n+je_2\,|\,j\in \N_0\}\;,\\
         & \mathtt{Q}_n\;:=\;\{n+je_1+ke_2\,|\,j,k\in \N_0\}\;,\\
          &\mathtt{Q}_n^{\rm c}\;:=\;\Z^2\setminus \mathtt{Q}_n\;,
    \end{split}
\end{equation}
where $\N_0:=\N\cup\{0\}$. Then, for every $n\in \Z^2$, the projections $\mathfrak{r}_n:=\mathfrak{m}_{\chi_{\mathtt{R}_n}}$, $\mathfrak{u}_n:=\mathfrak{m}_{\chi_{\mathtt{U}_n}}$, $\mathfrak{q}_n:=\mathfrak{m}_{\chi_{\mathtt{Q}_n}}$ and $\mathfrak{q}^\bot_n:=\mathfrak{m}_{\chi_{\mathtt{Q}^{\rm c}_n}}$  are elements of $\mathfrak{F}_B$. 
\end{proposition}
\begin{proof}
For any $m=(m_1,m_2)\in \Z^2$ one has that
\begin{equation*}
    \begin{split}
f_B(m+e_1+e_2)-f_B(m+e_1)\;
   &=\;\left\{ \begin{array}{lll}
              \expo{\ii b_{\llcorner}}- \expo{\ii b_{\ast}} &\quad&   \text{if}\,\,m_1\geq 0\;\;\text{and}\;\;m_2=0\;,\\
          0 &\quad& \text{otherwise}\;.
             \end{array}
   \right.
    \end{split}
\end{equation*}
Furthermore, $\expo{\ii b_\llcorner}-\expo{\ii b_{\ast}}\neq 0$ in view of $b_\llcorner-b_\ast\notin 2\pi \Z$ and this with the above yield
$$
\mathfrak{r}_0\;=\;(\expo{\ii b_\llcorner}-\expo{\ii b_{\ast}})^{-1}\big(\tau_{(-e_1-e_2)}(\mathfrak{f}_B)-\tau_{(-e_1)}(\mathfrak{f}_B)\big)\;\in\;\mathfrak{F}_B$$
and $\mathfrak{r}_n=\tau_n(\mathfrak{r}_0)\in\mathfrak{F}_B$ for every
$n\in \Z^2$. In a similar way, one has that
\begin{equation*}
    \begin{split}
f_B(m+e_1+e_2)-f_B(m+e_2)\;&=\;\left\{ \begin{array}{lll}
              \expo{\ii b_{\llcorner}}- \expo{\ii b_{\ast}} &\quad&   \text{if}\,\,m_1= 0\;\;\text{and}\;\;m_2\geq0\;,\\
          0 &\quad&  \text{otherwise}\;,
             \end{array}
   \right.
    \end{split}
\end{equation*}
and in turn
$$
\mathfrak{u}_{0}\;=\;(\expo{\ii b_\llcorner}-\expo{\ii b_{\ast}})^{-1}\big(\tau_{(-e_1-e_2)}(\mathfrak{f}_B)-\tau_{(-e_2)}(\mathfrak{f}_B)\big)\;\in\; \mathfrak{F}_B\;.
$$
As a consequence $\mathfrak{u}_n=\tau_n(\mathfrak{u}_0)\in\mathfrak{F}_B$. Finally
\begin{equation*}
    \begin{split}
f_B(m+e_1+e_2)-\expo{\ii b_\ast}\;&=\;\left\{ \begin{array}{lll}
              0 &\quad &  \text{if}\,\,m_1\leq 0\;\;\text{or}\;\;m_2\leq0\;,\\
          \expo{\ii b_\llcorner}-\expo{\ii b_\ast} &\quad&  \text{otherwise}\;,
             \end{array}
   \right.
    \end{split}
\end{equation*}
which implies that
$$
\mathfrak{q}_{0}\;=\;(\expo{\ii b_\llcorner}-\expo{\ii b_{\ast}})^{-1}\big(\tau_{(-e_1-e_2)}(\mathfrak{f}_B)-\expo{\ii b_{\ast}}\mathfrak{1}\big)\;\in\; \mathfrak{F}_B\;,
$$
and in turn $\mathfrak{q}_n=\tau_n(\mathfrak{q}_0)\in\mathfrak{F}_B$ for every
$n\in \Z^2$. To conclude the proof it is sufficient to observe that $\mathfrak{q}^\bot_n=\mathfrak{1}-\mathfrak{q}_{n}$.
\end{proof}

\begin{remark}
Let us say a few words about the notation introduced above. The symbols $\mathtt{R}_n$ and $\mathtt{U}_n$ stand for the \virg{right half-line} and the \virg{up half-line} at the point $n$, respectively. In the same vein, the symbol $\mathtt{Q}_n$ stands for the \virg{quarter-plane} at the point $n$.
\end{remark}

\medskip

Let $C_0(\Z^2)$ be the $C^*$-algebra of continuous functions which vanish at infinity and $C_\infty(\Z^2)=C_0(\Z^2)+\C\mathfrak{1}$ its unitization, \ie, the set of continuous functions with a limit at infinity. 
Due to the discreteness of $\Z^2$, every function on $\Z^2$ is automatically continuous. This provides the identification $\ell^\infty(\Z^2)\equiv C_b(\Z^2)$, where  $C_b$ denotes the set of bounded continuous functions.  Let us define
$$
\mathfrak{M}_\infty\;:=\;\left\{\mathfrak{m}_g\;|\; g\in C_\infty(\Z^2)\right\}\;,\qquad \mathfrak{M}^\infty\;:=\;\left\{\mathfrak{m}_g\;|\; g\in \ell^\infty(\Z^2)\right\}\;.
$$
Both $\mathfrak{M}_\infty$ and $\mathfrak{M}^\infty$ are unital sub-$C^*$-algebras of $\mathcal{B}(\ell^2(\Z^2))$ invariant under 
the $\Z^2$-action implemented by the translations $\tau$. It holds true that
\begin{equation}\label{eq:inclu}
\mathfrak{M}_\infty\;\subset\;\mathfrak{F}_B\;\subset\;\mathfrak{M}^\infty\;.
\end{equation}
The second inclusion in \eqref{eq:inclu} follows from the observation that $\mathfrak{f}_B\in\mathfrak{M}^\infty$. For the first inclusion let us observe that $\mathfrak{z}_n:=\mathfrak{r}_n\mathfrak{u}_n$ is the projection on the site $n$, \ie, 
the multiplication operator for the characteristic function $\chi_{\{n\}}$, for every $n\in\Z^2$. Therefore, since 
$\mathfrak{z}_n\in \mathfrak{M}_\infty\cap\mathfrak{F}_B$ and the $\mathfrak{z}_n$ generates $\mathfrak{M}_\infty$, one deduces the first inclusion of \eqref{eq:inclu}. In view of the  Gelfand-Naímark isomorphism, one gets that  
$\mathfrak{M}_\infty\simeq C(\Z^2_\infty)$ and $\mathfrak{M}^\infty\simeq C(\beta\Z^2)$ where $\Z^2_\infty$ and $\beta\Z^2$ are the \emph{one-point compactification} and the \emph{Stone-\v{C}ech compactification} of $\Z^2$, respectively \cite[Section 1.3]{gracia-varilly-figueroa-01}. Therefore $\Omega_B$ represents a compactification of $\Z^2$ which is between the one-point compactification and the Stone-\v{C}ech compactification. 
It is worth noticing that the setting we are considering has been studied in a certain generality in
\cite{georgescu-iftimovici-02,georgescu-iftimovici-06}.  In the following, we will provide a concrete model for $\Omega_B$.

\medskip

Let  $\infty_\ast,\infty_\llcorner,\infty_j^\mathtt{U},\infty_k^\mathtt{R}$ be given points with $j,k\in \Z$,  and define
$$\Omega_B\;:=\;\Z^2\;\cup\;\{\infty_\ast\}\;\cup\;\{\infty_\llcorner\}\;\cup\;\mathtt{U}_\infty\;\cup\;\mathtt{R}_\infty\;.
$$
where
$$
\mathtt{U}_\infty\;:=\;\bigcup_{j\in \Z}\; \{\infty_j^\mathtt{U}\}\;,\qquad \mathtt{R}_\infty\;:=\;\bigcup_{k\in \Z}\; \{\infty_k^\mathtt{R}\}\;.
$$
In order to introduce a convenient  topology on $\Omega_B$ we will first provide a basis for it. Let us introduce some convenient notation:
\begin{equation}
\begin{split}
   & \overline{\mathtt{R}_n}\;:=\;\mathtt{R}_n\;\cup\; \{\infty_{n_2}^\mathtt{R}\}\;,\\ &\overline{\mathtt{U}_n}\;:=\;\mathtt{U}_n\;\cup\; \{\infty_{n_1}^\mathtt{U}\}\;,\\
   &\overline{\mathtt{Q}_n}\;:=\;\mathtt{Q}_n\;\cup\;\{\infty_\llcorner\}\;\bigcup_{j\geq n_1}\;\{\infty_j^\mathtt{U}\}\;\bigcup_{j\geq n_2}\;\{\infty_j^\mathtt{R}\}\;,\\
   &\overline{\mathtt{Q}_n^{\rm c}}\;:=\;\mathtt{Q}_n^{\rm c}\;\cup\;\{\infty_\ast\}\;\bigcup_{j< n_1}\;\{\infty_j^\mathtt{U}\}\;\bigcup_{j< n_2}\;\{\infty_j^\mathtt{R}\},
\end{split}
\end{equation}
where $n=(n_1,n_2)\in \Z^2.$ It is  straightforward to check
that the collection
$$
\bb{Q}\;:=\;\left\{\{n\},\overline{\mathtt{R}_n},\overline{\mathtt{U}_n},\overline{\mathtt{Q}_n},\overline{\mathtt{Q}_n^{\rm c}}\;\big|\; n\in\Z^2\right\}\;,
$$
meets the conditions to be the basis of a topology for
 $\Omega_B$. We shall denote with
 $\mathtt{Top}_\bb{Q}$ to the topology on $\Omega_B$ generated by $\bb{Q}$.
\begin{proposition}
The topological space
$(\Omega_B, \mathtt{Top}_\bb{Q})$ is compact 
and $\Z^2$ sits inside $\Omega_B$ as an open dense subset. 
\end{proposition}
\begin{proof}
It is easy to see that $\Z^2$ is an open dense subset of $\Omega_B$ by the definition of $\mathtt{Top}_\bb{Q}$, hence let us prove the compactness.
Let $\{\mathtt{A}_\alpha\}_{\alpha\in \mathcal{I}}$ be an open cover of $\Omega_B$. Note that there exist $\alpha_\ast,\alpha_\llcorner\in \mathcal{I}$ such that $\infty_\ast\in \mathtt{A}_{\alpha_\ast}$ and $\infty_\llcorner\in \mathtt{A}_{\alpha_\llcorner}$. This implies that $\overline{\mathtt{Q}_n}\subset \mathtt{A}_{\alpha_\llcorner}$ and $\overline{\mathtt{Q}_m^{\rm c}}\in \mathtt{A}_{\alpha_\ast}$ for some $m=(m_1,m_2)$ and $n=(n_1,n_2)$ in $\Z^2$, in view of the fact that $\mathtt{A}_{\alpha_\ast}$ and $\mathtt{A}_{\alpha_\llcorner}$ are open.
Observe that there are at most a finite number of  elements $\infty_j^\mathtt{U}$ and $\infty_k^\mathtt{R}$ in the complement of $\overline{\mathtt{Q}_n}\cup\overline{\mathtt{Q}_m^{\rm c}}$. Therefore there is a finite number of elements of the cover such that $\infty_j^\mathtt{U}\in\overline{\mathtt{U}_{(j,s_j)}}\subset \mathtt{A}_{\alpha_j}$ and $\infty_k^\mathtt{R}\in\overline{\mathtt{R}_{(r_k,k)}}\subset \mathtt{A}_{\alpha_k}$ with $s_j, r_k\in\Z$ suitable coordinates.
As a result, the complement of the union of the finite collection of open sets $\mathtt{A}_{\alpha_\llcorner}$, $\mathtt{A}_{\alpha_\ast}$, $ \mathtt{A}_{\alpha_j}$ and  $\mathtt{A}_{\alpha_k}$ is at most a finite set of points of $\Z^2$.
As a consequence,  $\Omega_B$ can be covered by a finite number of elements of $\{\mathtt{A}_\alpha\}_{\alpha\in \mathfrak{I}}$
showing that $(\Omega_B, \mathtt{Top}_\mathcal{B})$ is a compact topological space. 
 \end{proof}
 Let $X$ be a discrete topological space.  A filter on $X$ is a nonempty family $\mathcal{F}$ of subsets of $X$, which is closed under finite intersections, does not contain the empty set, and for any $A\in\mathcal{F}$ such that if $A\subset B\subset X$, then $B\in \mathcal{F}.$ If $\mathcal{F}$ and $\mathcal{G}$ are filters on $X$, then we said that $\mathcal{G}$ is finer than $\mathcal{F}$ if $\mathcal{F}\subset \mathcal{G}$. An ultrafilter on $X$ is a filter that is not properly contained in any other filter on $X$. The space $\beta X$ of all ultrafilters on $X$ is a compact space, actually, $\beta X$ is the Stone-\v{C}ech compactification of $X$ (see \cite{Hin, Sam} and references therein).
 \medskip
 
 Let $D$ be a topological space and $f\colon X\rightarrow D$ any function. For a filter $\mathcal{F}$  on $X$ we say that $\lim_{\mathcal{F}}f=y$ or $\mathcal{F}$-limit of $f$ is $y$ if and only if for any open neighborhood $U$ of $y$ we have $f^{-1}(U)\in \mathcal{F}.$ It is worth pointing out that $\mathcal{F}$-limits  coincide with the general concept of limit in a topological space for functions defined on $\beta X$ \cite[Theorem 3.46]{Hin}. Note also that if $\lim_{\mathcal{F}}f=y$, then it is easy to see that for any filter $\mathcal{G}$ finer than $\mathcal{F}$ it holds that $\lim_{\mathcal{G}}f=y$.
 
 \medskip
 
 Recall that $\Z^2\subset \Omega_B\subset \beta \Z^2$, hence we can describe the elements of $\partial \Omega_B:=\Omega_B\setminus \Z^2$ as suitable class of ultrafilters on $\Z^2$ \cite[Lemma 2.1]{georgescu-iftimovici-06}. For that, let us define the following collections
 \begin{equation*}
 \begin{split}
     \mathcal{U}_j\;&:=\;\{A\subset \Z^2\;|\;\mathtt{U}_{(j,n_1)}\subset A,\;{\rm for\;some\;}n_1\in \Z\}\\
     \mathcal{R}_j\;&:=\;\{A\subset \Z^2\,|\;\mathtt{R}_{(n_1,j)}\subset A,\;{\rm for\;some\;}n_1\in \Z\}\\
     \mathcal{Q}_\llcorner\;&:=\;\{A\subset \Z^2\,|\;\mathtt{Q}_{n}\subset A,\;{\rm for\;some\;}n\in \Z^2\}\\
      \mathcal{Q}_\ast\;&:=\;\{A\subset \Z^2\,|\;\mathtt{Q}^c_{n}\subset A,\;{\rm for\;some\;}n\in \Z^2\}
 \end{split}
 \end{equation*}
One can see that the foregoing collections are filters of $\Z^2$. Furthermore, the continuity of each $g\in C(\Omega_B)$ and the definition of $\mathtt{Top}_\bb{Q}$ imply that
\begin{equation}\label{limits}
    \begin{split}
      \lim_{\mathcal{U}_j}g&\;=\;g(\infty^\U_j)\;,\hspace{1cm} \lim_{\mathcal{R}_j}g\;=\;g(\infty^{\rz}_j)\;,\hspace{1cm}j\in \Z\;,\\
    \lim_{\mathcal{Q}_\llcorner}g&\;=\;g(\infty_\llcorner)\;,\hspace{1cm}
     \lim_{\mathcal{Q}_\ast}g\;=\;g(\infty_\ast)\;.
    \end{split}
\end{equation}
Actually, if $f\colon \Z^2\rightarrow \C$ and all the limits in \ref{limits} exist for $f$, then $f$ has unique continuous extension on $\Omega_B$. Indeed, let $\hat{f}$ be the extension of $f$ on $\Omega_B$ given by
$$\hat{f}(\omega)\;:=\;\left\{ \begin{array}{lcc}
             {f}(\omega) &   \;if\,\,\,\omega\in \Z^2\\
\lim_{\mathcal{U}_j}f           & \; if\,\,\,\omega=\infty^\mathtt{U}_j\\
\lim_{\mathcal{R}_j}f           & \; if\,\,\,\omega=\infty^\mathtt{R}_j\\
\lim_{\mathcal{Q}_\llcorner}f           & \; if\,\,\,\omega=\infty_\llcorner\\
\lim_{\mathcal{Q}_\ast}f          & \; if\,\,\,\omega=\infty_\ast
             \end{array}
   \right.$$
From the definition of the limit, we have that for any $\epsilon>0$ there is $A\in \mathcal{U}_j$ such that $|\hat{f}(\infty_j^\mathtt{U})-\hat{f}(x)|<\epsilon$ for all $x\in A.$ In particular, since $\mathtt{U}_n\subset A$ for some $n\in \Z^2$, then for all $x\in \mathtt{U}_n $ it is true that $|\hat{f}(\infty_j^\mathtt{U})-\hat{f}(x)|<\epsilon$, which implies that $\hat{f}$ is continuous in $\infty_j^\mathtt{U}.$
The continuity of $\hat{f}$ in the other points of $\partial \Omega_B$ follows with the same argument. It turns out that $\hat{f}\in C(\Omega_B)$ and this extension is unique due to that $\Z^2$ is a dense subset of $\Omega_B.$ Thus,
one concludes that the elements of $\partial\Omega_B$ can be identify as follows
   \begin{equation}\label{2,14}
       \begin{split}
           \infty_j^\U\;&\equiv\;\{ \mathcal{G}\in \beta \Z^2\;|\;\mathcal{U}_j\subset \mathcal{G}\}\;,\qquad
            \infty_j^\rz\;\equiv\;\{ \mathcal{G}\in \beta \Z^2\;|\;\mathcal{R}_j\subset \mathcal{G}\}\;,\\
             \infty_\llcorner\;&\equiv\;\{ \mathcal{G}\in \beta \Z^2\;|\;\mathcal{Q}_\llcorner\subset \mathcal{G}\}\;,\qquad
              \infty_\ast\;\equiv\;\{ \mathcal{G}\in \beta \Z^2\;|\;\mathcal{Q}_\ast\subset \mathcal{G}\}\;.
       \end{split}
   \end{equation}

 \begin{proposition}
 It holds true that
 $\mathfrak{F}_B\simeq C(\Omega_B)$ as $C^*$-algebras.
 \end{proposition}
 \proof
 Firstly, recall that we can identify the elements of $\mathfrak{F}_B$ with functions over $\Z^2$. In particular, $\mathfrak{f}_B$ is the multiplication by the  function $f_B$ over $\Z^2$. For any $\gamma\in \Z^2$, we see that  $\tau_\gamma({f}_B)$ is constant in the sets  $\mathtt{Q}_\gamma$ and $\mathtt{Q}^{\rm c}_\gamma$, therefore the extension of $\tau_\gamma({f}_B)$ to $\Omega_B$ is a well defined continuous function, because the limits in \ref{limits} exist for $\tau_\gamma(f_B
)$. Moreover, these limits are linear and multiplicative \cite{Hin}, hence any polynomial $p$ in the variables $\tau_{\gamma_1}(f_B),\ldots,\tau_{\gamma_k}(f_B)$ have a continuous extension to $\Omega_B$ for any $k\in \N$.
Now let $\mathfrak{m}_g\in \mathfrak{F}_B$ be a generic element in $\mathfrak{F}_B$ defined by the multiplication by the function $g\in\ell^\infty(\Z^2)$ and define the map $\mathfrak{F}_B\ni\mathfrak{m}_g\mapsto \hat{{g}}\in C(\Omega_B)$, where $\hat{g}$ is constructed as follows. For such $g$ there is a sequence $\big(g_n\big)_{n\in \N}$ of polynomials $g_n$ in the variables $\tau_{\gamma_1}(f_B),\ldots,\tau_{\gamma_n}(f_B)$
so that $g_n\rightarrow g$ in $\mathfrak{F}_B$. Note that the sequence $\big(\hat{g}_n\big)_{n\in \N}$ defined by the extension of each $g_n$ on $\Omega_B$ is a Cauchy sequence in $C(\Omega_B)$. Namely, for every $\epsilon >0$ there is $N\in \N$ such that 
$\|g_n-g_m\|_{\mathfrak{F}_B}<\epsilon$ for all $n,m>N.$ Therefore,
\begin{equation*}
    \begin{split}
        \|\hat{g}_n-\hat{g}_m\|_{C(\Omega_B)}\;&=\;\sup_{x\in \Z^2}|\hat{g}_n(x)-\hat{g}_m(x)|\;=\; \|g_n-g_m\|_{\mathfrak{F}_B}\;<\;\epsilon\;.
    \end{split}
\end{equation*}
The latter follows from the fact that $\Z^2$ is dense in the compact  set $\Omega_B$. Thus, we define $\hat{g}$ as the limit in $C(\Omega_B)$ of the sequence $\big(\hat{g}_n\big)_{n\in \N}$, which provides a continuous extension of $g$ on $\Omega_B$. Using again the fact that $\Z^2$ is an open dense subset of $\Omega_B$, then any such $g$ has a unique continuous extension on $\Omega_B$, that is, the map $\mathfrak{F}_B\ni\mathfrak{m}_g\mapsto \hat{{g}}\in C(\Omega_B)$ is injective. In order to show the surjectivity,
we claim that the image of $\mathfrak{F}_B$ separates the points of $\Omega_B$. Indeed, let $\omega_1$ and $\omega_2$ be two points in $\Omega_B$, then there is $n\in \Z^2$ such that either $\omega_1\in \overline{\mathtt{Q}_n}$ and $\omega_2\in \overline{\mathtt{Q}_n}^c$ or $\omega_1\in \overline{\mathtt{Q}_n}^c$ and $\omega_2\in \overline{\mathtt{Q}_n}$. Thus one obtains that $\hat{\chi}_{\mathtt{Q}_n}(\omega_1)\neq\hat{\chi}_{\mathtt{Q}_n}(\omega_2)$ and the claim follows, where $\hat{\chi}_{\mathtt{Q}_n}$ is the extension of the indicator function $\chi_{\mathtt{Q}_n}$ on $\Z^2$, which belong to $\mathfrak{F}_B$ by Proposition \ref{Prop projections}. Finally, one can verify that this map is also linear, multiplicative and preserve the $*$-involution. Consequently, $\mathfrak{F}_B\ni\mathfrak{m}_g\mapsto \hat{{g}}\in C(\Omega_B)$ is a $*$-isomorphism between  $\mathfrak{F}_B$ and $C(\Omega_B)$.

\qed

 As expected $\Omega_B$ is the Gelfand spectrum of $\mathfrak{F}_B$, that is, the set of $*$-homomorphism $\omega\colon\mathfrak{F}_B\rightarrow\C$. The action of $\Omega_B$ over $\mathfrak{F}_B$ can be regarded as an evaluation map. Namely, the inclusion $\Z^2\ni n\mapsto\omega_n\in \Omega_B$ is given by the evaluation at a finite distance:
$$\omega_n(\mathfrak{g})\;:=\;g(n)\;,\qquad\mathfrak{g}\;=\;\mathfrak{m}_g\;\in\;\mathfrak{F}_B$$ 
and the remaining limit points are identified in the following way
\begin{equation}
\begin{split}
   \omega_{\infty_j^\mathtt{R}}(\mathfrak{g})&\;:=\;\lim_{\mathcal{R}_j}g\;,\hspace{2cm} \omega_{\infty_\ast}(\mathfrak{g})\;:=\;\lim_{\mathcal{Q}_\ast}g\;,\\
    \omega_{\infty_j^\mathtt{U}}(\mathfrak{g})&\;:=\;\lim_{\mathcal{U}_j}g\;,\hspace{2cm} \omega_{\infty_\llcorner}(\mathfrak{g})\;:=\;\lim_{\mathcal{Q}_\llcorner}g\;.
\end{split}
\end{equation} 
Furthermore, we can endow $\Omega_B$ with a $\Z^2$-action: For each $\gamma\in \Z^2$ consider $\tau^*_\gamma(\omega(\mathfrak{g}))\;:=\;\omega(\tau_{-\gamma}(\mathfrak{g}))$, where $\omega\in \Omega_B$ and $\mathfrak{g}\in \mathfrak{F}_B.$ Therefore, it follows that

\begin{equation*}
    \begin{split}
\tau^*_\gamma(\omega_{n})\;&=\;\omega_{n+\gamma}\;,\hspace{1cm}        \tau^*_\gamma(\omega_{\infty_j^\U})\;=\;\omega_{\infty_{j+\gamma_1}^\U}\;,\hspace{1cm} \tau^*_\gamma(\omega_{\infty_j^\rz})\;=\;\omega_{\infty_{j+\gamma_2}^\rz}\;,\\
\tau^*_\gamma(\omega_{\llcorner})\;&=\;\omega_{\llcorner}\;,\hspace{1,5cm}        \tau^*_\gamma(\omega_{\infty_\ast})\;=\;\omega_{\infty_\ast}\;.
    \end{split}
\end{equation*}
for every $\gamma=(\gamma_1,\gamma_2)\in\Z^2$. As a consequence
 $\Z^2$ is an invariant subset under the action $\tau^*$, and the boundary
\begin{equation}
\partial\Omega_B\;=\;\Omega_B\setminus\Z^2\;=\;\{\infty_\ast\}\;\cup\;\{\infty_\llcorner\}\;\cup\;\mathtt{U}_\infty\;\cup\;\mathtt{R}_\infty\;
\end{equation}
is the disjoint union of four invariant subsets.
\begin{proposition}
$(\Omega_B, \tau^*,\Z^2)$ is a dynamical system and its set ${\rm Erg}(\Omega_B)$ of the ergodic probability measures is given by ${\rm Erg}(\Omega_B)=\{\mathbb{P}_{\infty_\llcorner},\,\mathbb{P}_{\infty_\ast}\}$, where the measures are specified by 
 $$\mathbb{P}_{\infty_\llcorner}(\infty_\llcorner)\;=\;1\;,\qquad \mathbb{P}_{\infty_\ast}(\infty_\ast)\;=\;1\;.$$
\end{proposition}
\begin{proof}
From the definition of $\tau^*$, the minimal invariant subsets of $\Omega_B$ are $\Z^2$, $\U_\infty$, $\rz_\infty$, $\{\infty_\llcorner\}$ and $\{\infty_\ast\}.$ Moreover, the sets  $\Z^2$, $\U_\infty$ and $\rz_\infty$ are made of wandering points under $\tau^*$, hence the only possible invariant measures on $\Omega_B$ are the Dirac measures supported in $\{\infty_\llcorner\}$ or $\{\infty_\ast\}$, that is, ${\rm Erg}(\Omega_B)=\{\mathbb{P}_{\infty_\llcorner},\,\mathbb{P}_{\infty_\ast}\}$.
\end{proof}
 The ergodic measures of $(\Omega_B, \tau^*,\Z^2)$ play a crucial role in the construction of the integration theory of magnetic algebra. Indeed, for each $\mathbb{P}\in {\rm Erg}(\Omega_B)$ there is a $\Z^2$-invariant trace-per-unit-volume $\mathcal{T}_\mathbb{P}$ on $\A_B$  \cite[Section 2.6]{Deni1}.

 \section{Magnetic interfaces and K-theory}
 In this section, we calculate the K-groups of the magnetic algebra $\A_B$. More information about K-theory can be found in the books \cite{Bla,Ro, Ols}, and for preliminary results related to the K-theory of magnetic algebras see \cite{Deni1, PRO}.
 Along the way, in some computations we use the Iwatsuka magnetic algebra $\A_I$ and its magnetic hull $\Omega_I$, so we refer to \cite[Section 4 and Appendix B]{Deni1} for the detailed description of $\A_I$ and $\Omega_I.$
\subsection{Interface and bulk algebra}\label{secc 3.1}
Let $A_{B_1}$ and $A_{B_2}$ be two vector potentials such that their circulation is the magnetic fields $B_1$ and $B_2$, respectively. Consider the magnetic algebras $\A_{A_{B_1}}$ and $\A_{A_{B_2}}$\footnote{In the sense of \cite[Definition 2.9]{Deni1}} associated to the vector potentials. According to \cite[Definition 3.1]{Deni1}, an evaluation homomorphism, when it exists, is a surjective $*$-homomorphism ${\rm ev}\colon\A_{A_{B_1}}\rightarrow\A_{A_{B_2}}$  which fulfills 
\begin{equation*}
    \begin{split}
        {\rm ev}(\mathfrak{s}_{A_{B_1},1})\;:&=\;\mathfrak{s}_{A_{B_2},1}\\
         {\rm ev}(\mathfrak{s}_{A_{B_1},2})\;:&=\;\mathfrak{s}_{A_{B_2},2}
    \end{split}
\end{equation*}
Moreover, any evaluation homomorphism satisfies the following equality
\begin{equation}\label{Translation}
   {\rm ev}\big(\tau_\gamma(\mathfrak{f}_{B_1})\big)\;=\;\tau_\gamma(\mathfrak{f}_{B_2})\;, \qquad\forall\,\gamma\in\Z^2.
\end{equation}
Here $\tau_\gamma$ denotes $\Z^2$-action by translation defined in \ref{equ 2.6}.
\medskip

Since the full system is composed of two different materials with constant magnetic fields $b_\llcorner$ and $b_\ast$, then according to \cite[Definition 3.1.1]{PRO}, it is natural to define the \emph{bulk algebra} as
$$\A_{{\rm bulk}}\;:=\;\A_{b_{\llcorner}}\oplus\A_{b_\ast},$$
where $\A_{b_{\llcorner}}$ and $\A_{b_\ast}$ are the magnetic $C^*$-algebras associated to the constant magnetic fields of strength $b_{\llcorner}$ and $b_\ast$, respectively \cite[Example 2.10]{Deni1}. In other words, the bulk algebra contains the information  deep in the interior of both materials.  
\medskip

For sake of notational simplicity, from now on we write $\A$ instead of $\A_B$.
In view of \cite[Proposition 3.11]{Deni1}, the map ${\rm ev}\colon \A\rightarrow\A_{{\rm bulk}}$ defined as

\begin{equation}\label{sec 3.1}
\begin{split}
     {\rm ev}(\mathfrak{s}_1)\;:&=\;(\mathfrak{s}_{b_{\llcorner},1},\,\mathfrak{s}_{b_\ast,1})\\
       {\rm  ev}(\mathfrak{s}_2)\;:&=\;(\mathfrak{s}_{b_{\llcorner},2},\,\mathfrak{s}_{b_\ast,2})\\
        {\rm ev}(\mathfrak{f}_B)\;:&=\;\big(\,\expo{\ii b_{\llcorner}}\mathfrak{1},\,\expo{\ii b_\ast}\mathfrak{1}\big)
\end{split}
\end{equation}
is an evaluation homomorphism. We denote by $\mathfrak{I}$ to the kernel of ${\rm ev}$ and we call this ideal the \emph{interface algebra}. It is important to point out that the interface coincides with the subalgebra of $\A$ generated by the elements $\mathfrak{a}\mathfrak{g}\mathfrak{b}$ with $\mathfrak{a,b}\in \A$ and $\mathfrak{g}\in {\rm Ker}({\rm ev}|_{\mathfrak{F}_B})$ \cite[Section 3.1]{Deni1}, namely
\begin{equation}\label{interface}
    \mathfrak{I}\;=\;\overline{{\rm Span}}\,\{\mathfrak{a}\mathfrak{g}\mathfrak{b}\;|\;\mathfrak{a,b}\in \A,\;\mathfrak{g}\in {\rm Ker}({\rm ev}|_{\mathfrak{F}_B})\}\;.
\end{equation}
 In the following we will obtain an explicit description of $\mathfrak{I}$ in terms of the projections $\mathfrak{r}_0$ and $\mathfrak{u}_0,$ which implies that $\mathfrak{I}$ describes the behavior of the system near the interface. 
\begin{proposition}\label{prop: proj}
The projections $\mathfrak{r}_n$ and $\mathfrak{u}_n$ belong to $\mathfrak{I}$ for every $n\in \Z^2$. Moreover, it holds true that for $n\in \Z^2$
\begin{equation}
\begin{split}
            {\rm ev}(\mathfrak{q}_n)\;&=\;(\mathfrak{1},\mathfrak{0})\\
             {\rm ev}(\mathfrak{q}^\bot_n)\;&=\;(\mathfrak{0},\mathfrak{1})
\end{split}
\end{equation}
\end{proposition}
\begin{proof}
Thanks to the proposition \ref{Prop projections}, we know that $\mathfrak{r}_n$, $\mathfrak{u}_n$ and $\mathfrak{q}_n$ can be written as
\begin{equation*}
    \begin{split}
        \mathfrak{r}_{n}\;&=\;(\expo{\ii b_\llcorner}-\expo{\ii b_{\ast}})^{-1}\big(\tau_{(n-e_1-e_2)}(\mathfrak{f}_B)-\tau_{(n-e_1)}(\mathfrak{f}_B)\big),\\
\mathfrak{u}_{n}\;&=\;(\expo{\ii b_\llcorner}-\expo{\ii b_{\ast}})^{-1}\big(\tau_{(n-e_1-e_2)}(\mathfrak{f}_B)-\tau_{(n-e_2)}(\mathfrak{f}_B)\big)\\
\mathfrak{q}_{n}\;&=\;(\expo{\ii b_\llcorner}-\expo{\ii b_{\ast}})^{-1}\big(\tau_{(n-e_1-e_2)}(\mathfrak{f}_B)-\expo{\ii b_{\ast}}\mathfrak{1}\big),
    \end{split}
\end{equation*}
for $n=(n_1,n_2)\in \Z^2$. Using \ref{Translation} and the fact that ${\rm ev}(\mathfrak{f}_B)$ is invariant under the action $\tau$ (constant magnetic field), then the above equations show that ${\rm ev}(\mathfrak{r}_n)=0={\rm ev}(\mathfrak{u}_n)$. The other equalities follow with the same argument.
\end{proof}

\begin{proposition}\label{prop 3,2}
The interface algebra is the closed two-sided ideal generated by $\mathfrak{r}_0$ and $\mathfrak{u}_0$.
\end{proposition}
\vspace{-1cm}
\begin{proof}
We write $\A(\mathfrak{r}_0,\mathfrak{u}_0)$ for the closed two-sided ideal generated by  $\mathfrak{r}_0$ and $\mathfrak{u}_0$ in $\A$. Firstly, the Proposition \ref{prop: proj} implies that
$\A(\mathfrak{r}_0,\mathfrak{u}_0)\subset\mathfrak{I}$ so we just have to show the reverse inclusion. Indeed, for every $n=(n_1,n_2)\in \Z^2$ we have the following decomposition
\begin{equation*}
    \begin{split}
        \tau_n(\fb)\;&=\;\fb+\big(\tau_n(\fb)-\tau_{(0,n_2)}(\fb)\big)+\big(\tau_{(0,n_2)}(\fb)-\fb\big)\\
        &=\;\fb+(\expo{\ii b_{\ast}}-\expo{\ii b_\llcorner})\sum_{m=1}^{n_1}\mathfrak{u}_{(m,n_2+1)}+(\expo{\ii b_{\ast}}-\expo{\ii b_\llcorner})\sum_{m=1}^{n_2}\mathfrak{r}_{(1,m)}\;,
    \end{split}
\end{equation*}
and since $\mathfrak{u}_n$ and $\mathfrak{r}_n$ belong to $\A(\mathfrak{r}_0,\mathfrak{u}_0)$  for any $n\in \Z^2$ we obtain that $\tau_n(\fb)$ is equal to $\fb$ modulo $\A(\mathfrak{r}_0,\mathfrak{u}_0)$. Moreover, the spectrum of $\mathfrak{f}_B$ is the set $\{ \expo{\ii b_\ast}, \expo{\ii b_\llcorner}\}$ and it yields $$\mathfrak{F}_B/\A(\mathfrak{r}_0,\mathfrak{u}_0)\; \simeq\; C^*(\fb)\;\simeq \;\C^2.$$ 
Thereby if $\mathfrak{g}\in \mathfrak{F}_B$ then $\mathfrak{g}=\lambda_1 \mathfrak{q}_0+\lambda_2\mathfrak{q}_0^\bot$ modulo $\A(\mathfrak{r}_0,\mathfrak{u}_0)$  for some $\lambda_1, \lambda_2\in \C$. Hence we conclude by Proposition \ref{prop: proj} that ${\rm ev}(\mathfrak{g})=(\mathfrak{0,0})$ if and only if $\lambda_1=\lambda_2=0$, that is, $\mathfrak{g}\in \A(\mathfrak{r}_0,\mathfrak{u}_0)$. The latter with \ref{interface} show the reverse inclusion.
\end{proof}

\subsection{K-theory of the magnetic hull}
Let $\Omega_B$ be the magnetic hull and consider the following exact sequence
\begin{equation}\label{seq 3.2}
 \xymatrix{
0\ar[r]^{} & C_0(\Z^2)\ar[r]^{i} & C(\Omega_B) \ar[r]^{e}& C(\partial\Omega_B)\ar[r]^{}&0 }
\end{equation}
where $C_0(\Z^2)$ is the $C^*$-algebra of sequences vanishing at infinity, $i$ is the inclusion homomorphism and $e$ is the evaluation homomorphism at the limit points. Notice that $\partial \Omega_B$ is homeomorphic to the Iwatsuka magnetic hull $\Omega_I$ \cite[Example 2.24]{Deni1} with the identifications $\infty_j^{\mathtt{U}}\mapsto 2j+1$, $\infty_j^{\mathtt{R}}\mapsto 2j$, $\infty_\llcorner\mapsto +\infty$ and $\infty_\ast\mapsto -\infty$, for $j\in \Z$. Therefore, the exact sequence \ref{seq 3.2} turns out to be \begin{equation}
 \xymatrix{
0\ar[r]^{} & C_0(\Z^2)\ar[r]^{} & C(\Omega_B) \ar[r]^{}& C(\Omega_I)\ar[r]^{}&0 }.
\end{equation}
The K-theory of the Iwatsuka magnetic hull is given by $K_0\big(C(\Omega_I)\big)=\Z^{\oplus\Z}\oplus\Z^2$ and $K_1\big(C(\Omega_I)\big)=0$
 \cite[Appendix B]{Deni1}. Furthermore,  $K_0\big(C_0(\Z^2)\big)=K^0(\Z^2)=\Z^{\oplus \Z^2}$ and $K_1\big(C_0(\Z^2)\big)=0.$ Thus, the six-term exact sequence implies that 
$$\xymatrix{
 \Z^{\oplus\Z^2}\ar[r]^{} & K_0\big(C(\Omega_B)\big) \ar[r]^{}& \Z^{\oplus\Z}\oplus\Z^2\ar[d] \\
0\ar[u] &K_1\big(C(\Omega_B)\big)\ar[l] &0 \ar[l] }$$
is exact. Since the groups in the latter sequence are free $\Z$-modules, then the extensions are trivial and one has that\\
\begin{equation*}
    \begin{split}
        K_0\big(C( \Omega_B)\big)\;&=\; i_*\Big(K_0\big(C_0(\Z^2)\big)\Big)\oplus e_*^{-1}\Big(K_0\big(C(\Omega_I)\big)\Big)\\&=\;\bigoplus_{i\in \Z^2}\Z[\mathfrak{z}_i]\oplus\bigoplus_{i\in \Z}\Z[\mathfrak{r}_{(0,i)}]\oplus\bigoplus_{i\in \Z}\Z[\mathfrak{u}_{(i,0)}]\oplus \Z[\mathfrak{q}^\bot_0]\oplus\Z[\mathfrak{q}_0]\end{split}
\end{equation*}
and
$K_1\big(C(\Omega_B)\big)=0,$ where recall that $(\mathfrak{z}_i\psi)(n):=\delta_{j,n}\psi(n)$ are the generators of $K_0\big(C_0(\Z^2)\big)$ and $e^{-1}_*$ is a splitting homomorphism.
\subsection{K-theory of the magnetic quarter-plane algebra}
It is well known that the magnetic algebra $\A$ has a crossed-product structure given by
\begin{equation}
    \A\;=\;\big(\mathfrak{F}_B\rtimes_{\alpha_1}\Z\big)\rtimes_{\alpha_2}\Z
\end{equation}
where the automorphism $\alpha_1$ is defined by $\alpha_1(\mathfrak{g}):=\mathfrak{s}_1\mathfrak{g}\mathfrak{s}_1^*$ for $\mathfrak{g}\in \mathfrak{F}_B$ and the automorphism $\alpha_2$ is given by $\alpha_2(\mathfrak{gs}^r_1):=\mathfrak{s}_2\mathfrak{gs}_1^r\mathfrak{s}_2^*$ for every $\mathfrak{g}\in \mathfrak{F}_B$ and $r\in \N_0$. A discussion of the above result can be found in \cite[Appendix A]{Deni1}, and for more information about the crossed product of $C^*$-algebras we refer to \cite{Bla, Dav, Ped, Wil}. 
\medskip

From the Pimsner-Voiculescu exact sequence \cite{Pim}, one can relate the K-theory of $\mathcal{Y}_1=\mathfrak{F}_B\rtimes_{\alpha_1}\Z$ with the $K$-groups of $\mathfrak{F}_B$ as
\begin{equation}\label{seq pim}
    \xymatrix{
 K_0\big(\mathfrak{F}_B\big)\ar[r]^{\beta_{1*}} & K_0\big(\mathfrak{F}_B\big) \ar[r]^{i_*}& K_0\big(\mathcal{Y}_1\big)\ar[d]^{\partial_0} \\
K_1\big(\mathcal{Y}_1\big)\ar[u]^{\partial_1} &K_1\big(\mathfrak{F}_B\big)\ar[l]^{i_*} &K_1\big(\mathfrak{F}_B\big) \ar[l]^{\beta_{1*}} }
\end{equation}
Here the vertical maps $\partial_0$ and $\partial_1$ are the index and exponential maps of a suitable six-term exact sequence \cite[Chapter V]{Bla}, and $\beta_1:= {\bf1}-\alpha_1$.
Therefore, replacing the known $K$-groups in \ref{seq pim} it turns out that
$$   \xymatrix{
\Z^{\oplus \Z^2}\oplus\Z^{\oplus \Z}\oplus \Z^2\ar[r]^{} & \Z^{\oplus \Z}\oplus\Z^{\oplus \Z}\oplus \Z^2 \ar[r]^{}& K_0\big(\mathcal{Y}_1\big)\ar[d] \\
K_1\big(\mathcal{Y}_1\big)\ar[u] &0\ar[l] &0 \ar[l] }$$
The remaining K-groups can be computed through the following proposition.
\begin{proposition}
The image and the kernel of the map $\beta_{1*}\colon K_0(\mathfrak{F}_B)\rightarrow K_0(\mathfrak{F}_B)$ are given by
$${\rm Im}(\beta_{1*})\;=\;\bigoplus_{i\in \Z^2}\Z[\mathfrak{z}_i]\oplus\bigoplus_{i\in \Z}\Z[\mathfrak{u}_{(i,0)}]\;,\qquad {\rm Ker}(\beta_{1*})\;=\;\Z[\mathfrak{1}]\;.$$
Therefore,
$$K_0(\mathcal{Y}_1)\;=\;
\bigoplus_{i\in \Z}\Z[\mathfrak{r}_{(0,i)}]\oplus \Z[\mathfrak{q}_0]\oplus\Z[\mathfrak{q}^\bot_0]\;,\qquad K_1(\mathcal{Y}_1)\;=\;\Z[\mathfrak{s}_1]\;.$$
\end{proposition}
\begin{proof}
Using the relations between the elements of $\mathfrak{F}_B$ one gets
\begin{equation*}
\begin{split}
      \beta_{1*}\big([\mathfrak{z}_{(i,j)}]\big)&\;=\;[\mathfrak{z}_{(i,j)}-\mathfrak{s}_1\mathfrak{z}_{(i,j)}\mathfrak{s}^*_1]\;=\;[\mathfrak{z}_{(i,j)}]-[\mathfrak{z}_{(i+1,j)}]\\
       \beta_{1*}\big([\mathfrak{q}_0]\big)&\;=\;[\mathfrak{q}_0-\mathfrak{s}_1\mathfrak{q}_0\mathfrak{s}_1^*]\;=\;[\mathfrak{u}_{(0,0)}]\\
        \beta_{1*}\big([\mathfrak{q}^\bot_0]\big)&\;=\;[\mathfrak{q}^\bot_0-\mathfrak{s}_1\mathfrak{q}^\bot_0\mathfrak{s}_1^*]\;=\;-[\mathfrak{u}_{(0,0)}]\\
 \beta_{1*}\big([\mathfrak{r}_{(0,i)}]\big)&\;=\;[\mathfrak{r}_{(0,i)}-\mathfrak{s}_1\mathfrak{r}_{(0,i)}\mathfrak{s}_1^*]\;=\;[\mathfrak{z}_{(0,i)}]\\
    \beta_{1*}\big([\mathfrak{u}_{(i,0)}]\big)
            &\;=\;[\mathfrak{u}_{(i,0)}-\mathfrak{s}_1\mathfrak{u}_{(i,0)}\mathfrak{s}_1^*]\;=\;[\mathfrak{u}_{(i,0)}]-[\mathfrak{u}_{(i+1,0)}]
\end{split}
\end{equation*}
Hence the image of $\beta_{1*}$ is $$\bigoplus_{i\in \Z^2}\Z[\mathfrak{z}_i]\oplus\bigoplus_{i\in \Z}\Z[\mathfrak{u}_{(i,0)}]\;.$$
The above relations also imply that the kernel of $\beta_{1*}$ is $ \Z[\mathfrak{q}_0+\mathfrak{q}^\bot_0]=\Z[\mathfrak{1}].$ Notice that the sequence \ref{seq pim} yields that ${\rm Ker}(\partial_1)={\rm Im}(i_*)=0$, then $$K_1(\mathcal{Y}_1)\;=\;{\rm Im}(\partial_1)\;=\;{\rm Ker}(\beta_{1*})\;\simeq\; \Z[\mathfrak{1}]\;.$$
By using the isometry $v_1:=\mathfrak{s}_1\otimes \mathfrak{v}$ defined in the proof of \cite[Proposition 4.9]{Deni1}, one obtains that $\partial_1([\mathfrak{s}_1])=-[\mathfrak{1}]$ and therefore $K_1(\mathcal{Y}_1)\;=\;\mathbb{Z}[\mathfrak{s}_1].$ The remaining $K$-group is given by
\begin{equation*}
\begin{split}
     K_0(\mathcal{Y}_1)&\;=\;K_0(\mathfrak{F}_B)/{\rm Im}(\beta_{1*})\;=\;\bigoplus_{i\in \Z}\Z[\mathfrak{r}_{(0,i)}]\oplus \Z[\mathfrak{q}_0]\oplus\Z[\mathfrak{q}^\bot_0]\;.
\end{split}
\end{equation*}
\end{proof}
\noindent
In order to obtain the $K$-theory of $\A$ we use again the Pimsner-Voiculescu exact sequence. Namely,
\begin{equation}\label{seq 3,8}
    \xymatrix{
 K_0\big(\mathcal{Y}_1\big)\ar[r]^{\beta_{2*}} & K_0\big(\mathcal{Y}_1\big) \ar[r]^{i_*}& K_0\big(\A\big)\ar[d]^{\partial_0} \\
K_1\big(\A\big)\ar[u]^{\partial_1} &K_1\big(\mathcal{Y}_1\big)\ar[l]^{i_*} &K_1\big(\mathcal{Y}_1\big) \ar[l]^{\beta_{2*}} }
\end{equation}
where $\beta_2:= {\bf1}-\alpha_2$. Replacing the known $K$-groups it follows that
$$   \xymatrix{
\Z^{\oplus \Z}\oplus \Z^2\ar[r]^{} & \Z^{\oplus \Z}\oplus \Z^2 \ar[r]^{}& K_0\big(\A\big)\ar[d] \\
K_1\big(\A\big)\ar[u] &\Z\ar[l] &\Z \ar[l] }.$$
Now we are ready to calculate the $K$-groups of the magnetic quarter-plane algebra $\A$.
\newpage
\begin{proposition}\label{k-groups A}
The image and the kernel of the map $\beta_{2*}\colon K_0(\mathcal{Y}_1)\rightarrow K_0(\mathcal{Y}_1)$ are given by
$${\rm Im}(\beta_{2*})\;=\;\bigoplus_{i\in \Z}\Z[\mathfrak{r}_{(0,i)}]\;,\qquad {\rm Ker}(\beta_{2*})\;=\;\Z[\mathfrak{1}]\;.$$
Therefore,
$$K_0(\A)\;=\;\Z[\mathfrak{q}_0]\oplus\Z[\mathfrak{q}^\bot_0]\oplus \Z[\mathfrak{c}]
\;,\qquad K_1(\A)\;=\;\Z[\mathfrak{s}_1]\oplus \Z[\mathfrak{s}_2]\;,$$
where $\mathfrak{c}$ is a projection in $\A\otimes {\rm Mat}_N(\C)$ for some $N\in \mathbb{N}.$
\end{proposition}
\begin{proof}
Using again the relations between the elements of $\mathfrak{F}_B$ we have 
\begin{equation*}
\begin{split}
       \beta_{2*}\big([\mathfrak{q}_0]\big)&\;=\;[\mathfrak{q}_0-\mathfrak{s}_2\mathfrak{q}_0\mathfrak{s}_2^*]\;=\;[\mathfrak{r}_{(0,0)}]\\
        \beta_{2*}\big([\mathfrak{q}^\bot_0]\big)&\;=\;[\mathfrak{q}^\bot_0-\mathfrak{s}_2\mathfrak{q}^\bot_0\mathfrak{s}_2^*]\;=\;-[\mathfrak{r}_{(0,0)}]
        \\
        \beta_{2*}\big([\mathfrak{r}_{(0,i)}]\big)&\;=\;[\mathfrak{r}_{(0,i)}-\mathfrak{s}_2\mathfrak{r}_{(0,i)}\mathfrak{s}_2^*]\;=\;[\mathfrak{r}_{(0,i)}]-[\mathfrak{r}_{(0,i+1)}]
\end{split}
\end{equation*}
It follows that  
$${\rm Im}(\beta_{2*})\;=\;\bigoplus_{i\in \Z}\Z[\mathfrak{r}_{(0,i)}]\;,\qquad {\rm Ker}(\beta_{2*})\;=\;\Z[\mathfrak{1}]\;,$$
and
\begin{equation}\label{kkk}
    \begin{split}
        K_1(\A)/{\rm Im}(i_*)\;\simeq\; {\rm Im}(\partial_1)\;\simeq\; {\rm Ker}(\beta_{2*})\;=\;\mathbb{Z}[\mathfrak{1}]\;.
    \end{split}
\end{equation}
Notice that
$$ \beta_{2*}\big([\mathfrak{s}_{1}]\big)\;=\;[\mathfrak{s}_1]-[\mathfrak{s_2\mathfrak{s}_1\mathfrak{s}_2^*}]\;=\;[\mathfrak{s}_1]-[\overline{\mathfrak{f}_B}\mathfrak{s}_1]\;=\;[\mathfrak{0}]\;,$$
where we have used that $[\overline{\mathfrak{f}_B}]=[\mathfrak{1}]$ with the homotopy $t\mapsto \expo{-\ii b_\llcorner t}\mathfrak{q}_0+\expo{-\ii b_\ast t}\mathfrak{q}_0^\bot.$
Hence $\beta_{2*}\colon\Z[\mathfrak{s}_1]\rightarrow\Z[\mathfrak{s}_1]$ is the zero map and consequently ${\rm Im}(i_*)\simeq \Z[\mathfrak{s}_1]$. Moreover, using again the same argument in the proof of \cite[Proposition 4.9]{Deni1} with the isometry $v_2=\mathfrak{s}_2\otimes \mathfrak{v}$, it holds that $\partial_1([\mathfrak{s}_2])=-[\mathfrak{1}]$ and so one concludes that $K_1(\A)=\Z[\mathfrak{s}_1]\oplus \Z[\mathfrak{s}_2]$ by \ref{kkk}. \\
Now let us calculate $K_0(\A)$. In light of the sequence \ref{seq 3,8}, one has
\begin{equation}\label{nnn}
    K_0(\A)/{\rm Im}(i_*)\;\simeq \;{\rm Im}(\partial_0)\;\simeq \;{\rm Ker}(\beta_{2*})\;=\;\mathbb{Z}[\mathfrak{s}_1]\;,
\end{equation}
where recall that $\beta_{2*}\colon\Z[\mathfrak{s}_1]\rightarrow\Z[\mathfrak{s}_1]$ is the zero map. Since $${\rm Ker}(i_*)\;=\;{\rm Im}(\beta_{2*})\;=\;\bigoplus_{i\in \Z}\Z[\mathfrak{r}_{(0,i)}]\;, $$
where $\beta_{2*}\colon K_0(\mathcal{Y}_1)\rightarrow K_0(\mathcal{Y}_1)$, then
$$i_*\big(\bigoplus_{i\in \Z}\Z[\mathfrak{r}_{(0,i)}]\oplus \Z[\mathfrak{q}_0]\oplus\Z[\mathfrak{q}^\bot_0]\big)\;=\;i_*\big( \Z[\mathfrak{q}_0]\oplus\Z[\mathfrak{q}^\bot_0]\big)\;\simeq \;\Z[\mathfrak{q}_0]\oplus\Z[\mathfrak{q}^\bot_0]\;.$$
Finally, for some $N\in \N$ there is a projection $\mathfrak{c}\in\mathfrak{A}\otimes {\rm Mat}_N(\C)$ such that $\partial_0([\mathfrak{c}])=-[\mathfrak{s}_1]$ and by \ref{nnn} it holds that $$K_0(\A)\;=\;\Z[\mathfrak{q}_0]\oplus\Z[\mathfrak{q}^\bot_0]\oplus \Z[\mathfrak{c}]\;.$$
\end{proof}
\subsection{K-theory of the interface algebra} As a direct consequence of the definition of the interface algebra $\mathfrak{I}$, we obtain the \emph{quarter-plane exact sequence}:
\begin{equation}\label{seq a}
    \xymatrix{
 0\ar[r]&\mathfrak{I}\ar[r]^{i} & \A \ar[r]^{{\rm ev}}& \A_{{\rm bulk}}\ar[r]&0\;, }
\end{equation}
 where recall that $\mathfrak{I}$ and $\A_{{\rm bulk}}$ are defined in section \ref{secc 3.1}.  Applying the six-term exact sequence to \ref{seq a} it turns out

\begin{equation}\label{seq six}
    \xymatrix{
 K_0(\mathfrak{I})\ar[r]^{i_*} & K_0(\A) \ar[r]^{{\rm ev}_*}& K_0(\A_{bulk})\ar[d]^{{\mathtt{exp}}} \\
K_1(\A_{bulk})\ar[u]^{{\mathtt{ind}}} &K_1(\A) \ar[l]^{{\rm ev}_*} &K_1(\mathfrak{I}) \ar[l]^{i_*} }
\end{equation}
Observe that the bulk algebra $\A_{{\rm bulk}}$ is isomorphic to the direct sum of two noncommutative torus \cite[Example 2.10]{Deni1}, and its K-theory is well-known \cite{Bla,Ols}. Namely,
$$K_0(\A_{{\rm bulk}})\;=\;\Z[(\mathfrak{1,0})]\oplus\Z[(\mathfrak{0,1})]\oplus\Z[(\mathfrak{p}_{\theta_\llcorner},\mathfrak{0})]\oplus\Z[(\mathfrak{0},\mathfrak{p}_{\theta_\ast})]\;,$$
$$K_1(\A_{{\rm bulk}})\;=\;\Z[(\mathfrak{s}_{b_{\llcorner},1},\mathfrak{1})]\oplus\Z[(\mathfrak{s}_{b_{\llcorner},2},\mathfrak{1})]\oplus\Z[(\mathfrak{1},\mathfrak{s}_{b_\ast,1})]\oplus\Z[(\mathfrak{1},\mathfrak{s}_{b_\ast,2})]\;,$$
where $\mathfrak{p}_{\theta_\llcorner}$ and $\mathfrak{p}_{\theta_\ast}$ are the \emph{Powers-Rieffel projections} of $\A_{b_\llcorner}$ and $\A_{b_\ast}$, respectively.
Thus, the sequence \ref{seq six} turns out to be
$$ \xymatrix{
 K_0(\mathfrak{I})\ar[r]^{i_*} & \Z^3 \ar[r]^{{\rm ev}_*}& \Z^4\ar[d]^{{\mathtt{exp}}} \\
\Z^4\ar[u]^{{\mathtt{ind}}} &\Z^2
\ar[l]^{{\rm ev}_*} &K_1(\mathfrak{I}) \ar[l]^{i_*} }$$
The following Theorem provides the K-theory of $\mathfrak{I}.$
\begin{theorem}
The K-groups of the interface algebra $\mathfrak{I}$ are given by
\begin{equation*}
\begin{split}
     K_0(\mathfrak{I})\;=\;\Z[\mathfrak{r}_{0}]\oplus\Z[\mathfrak{u}_{0}]\;,\qquad
    K_1(\mathfrak{I})\;=\;\Z[\mathfrak{w}]\;,
\end{split}
\end{equation*}
where $\mathfrak{w}\;:=\;(\mathfrak{s}_{1}-\mathfrak{1})\mathfrak{r}_{0}+(\mathfrak{s}_2^*-\mathfrak{1})\mathfrak{u}_{0}(\mathfrak{1}-\mathfrak{r}_0)+ \mathfrak{1}.$
\end{theorem}
\begin{proof}
Let us denote by $\A(\mathfrak{r}_0)$ and $\A(\mathfrak{u}_0)$ to the closed two-sided ideals in $\A$ generated by $\mathfrak{r}_{0}$ and $\mathfrak{u}_{0}$, respectively.
Adapting step by step the proof of \cite[Lemma 3,4]{Guo}, we obtain that $\mathcal{K}= \A(\mathfrak{r}_0)\cap\A(\mathfrak{u}_0)$, where $\mathcal{K}$ is the algebra of compact operators on $\ell^2(\Z^2)$. The latter allows to see that $$K_0\big(\A(\mathfrak{r}_0)\cap\A(\mathfrak{u}_0)\big)\;=\;\Z[\mathfrak{z}_{0}]\;,\qquad K_1\big(\A(\mathfrak{r}_0)\cap\A(\mathfrak{u}_0)\big)\;=\;0\;.$$ 
In order to calculate the K-theory of $\mathfrak{I}$, we claim that $$\A(\mathfrak{u}_0)/\A(\mathfrak{r}_0)\cap\A(\mathfrak{u}_0)\;\simeq\;\A(\mathfrak{r}_0)/\A(\mathfrak{r}_0)\cap\A(\mathfrak{u}_0)\;\simeq\; C(\mathbb{T})\otimes\mathcal{K}.$$
where $\mathbb{T}$ is the one-dimensional torus and, with a little abuse of notation, $\mathcal{K}$ is the algebra of compact operators on $\ell^2(\Z).$
In fact, consider the closed invariant subset $\U_\infty\cup\{\infty_\llcorner\}\cup\{\infty_\ast\}\subset \partial \Omega_B$, which is homeomorphic to the magnetic hull of the Iwatsuka magnetic algebra $\Omega_I$ with the identification $\infty^\U_j\mapsto j$, $\infty_\llcorner\mapsto +\infty$, and $\infty_\ast\mapsto-\infty$.  Let us write $\A_\U$ for the magnetic algebra with associated magnetic hull $\U_\infty\cup\{\infty_\llcorner\}\cup\{\infty_\ast\}$ in the sense of \cite[Proposition 3.11]{Deni1}.
Therefore, the associate evaluation map ${\rm ev}_\U:\A\rightarrow \A_\U$ is a surjective $*$-homomorphism and fulfills
\begin{equation*}
\begin{split}
    {\rm ev}_\U(\mathfrak{s}_1)\;=\;\mathfrak{s}_{1}^\U\;,\qquad
       {\rm ev}_\U(\mathfrak{s}_2)\;=\;\mathfrak{s}_{2}^\U\;,\qquad
        {\rm ev}_\U(\mathfrak{f}_B)\;=\;\mathfrak{f}_I^\U\;.
\end{split}
\end{equation*}
Since $\tau_{(0,n_2)}({\rm ev}(\mathfrak{f}_B))={\rm ev}(\mathfrak{f}_B)$ for any $n_2\in \Z$, then the kernel of ${\rm ev}_\U$ is $\A(\mathfrak{r}_0).$ Namely, the decomposition given in the proof of the proposition \ref{prop 3,2} together with the equality ${\rm ev}_\U\big(\tau_\gamma(\mathfrak{f}_B)\big)=\tau_{(\gamma_1,0)}(\mathfrak{f}_I^\U)$ for each $\gamma=(\gamma_1,\gamma_2)\in \Z^2$
show that $\A(\mathfrak{r}_0)={\rm Ker}({\rm ev}_\U).$ On the other hand, we know that ${\rm ev}_\U\big(\mathfrak{u}_0\big)=\mathfrak{p}_{0}^\U$, where $\mathfrak{p}_{0}^\U$ is the projection which generates the Iwatsuka interface \cite[Proposition 4,6]{Deni1}. Thus, the image of $\A(\mathfrak{u}_0)$ under ${\rm ev}_\U$ is $*$-isomorphic to the Iwatsuka interface, and
 from \cite[Proposition 4.2]{Deni1}, one has $$U{\rm ev}_\U(\A(\mathfrak{u}_0))U^{-1}\;=\; C(\mathbb{T})\otimes \mathcal{K}$$
where $U$ is the Bloch-Floquet transform \cite{Kuch}. Therefore, 
\begin{equation*}
    \begin{split}
        \A(\mathfrak{u}_0)/\A(\mathfrak{r}_0)\cap\A(\mathfrak{u}_0)\;&=\;\A(\mathfrak{u}_0)/{\rm Ker}({\rm ev}_\U|_{\A(\mathfrak{u}_0)})\;\simeq\; {\rm ev}_\U(\A(\mathfrak{u}_0))\\
        &\simeq\; C(\mathbb{T})\otimes\mathcal{K}\;.
    \end{split}
\end{equation*}
 Using now the closed invariant subset $\rz_\infty\cup\{\infty_\llcorner\}\cup\{\infty_\ast\}\subset \partial \Omega_B$, then the same argument yields
 \begin{equation*}
    \begin{split}
        \A(\mathfrak{r}_0)/\A(\mathfrak{r}_0)\cap\A(\mathfrak{u}_0)\;\simeq\; C(\mathbb{T})\otimes\mathcal{K}\;,
    \end{split}
\end{equation*}
and the claim follows. By Proposition \ref{prop 3,2}, $\mathfrak{I}=\A(\mathfrak{r}_0)+\A(\mathfrak{u}_0)$ and hence the claim yields the following exact sequence
\begin{equation}\label{10}
\xymatrix{0\ar[r]& \A(\mathfrak{r}_0)\cap\A(\mathfrak{u}_0)\ar[r]^{\qquad i}& \mathfrak{I}\ar[r]^{\hspace{-1.6cm}\pi}& C(\mathbb{T})\otimes\mathcal{K}\oplus C(\mathbb{T})\otimes\mathcal{K}\ar[r]&0.}
\end{equation}
The six-term exact sequence associated with \ref{10} implies
\begin{equation}\label{seq 3.12}
    \xymatrix{
 \Z\ar[r]^{i_*} & K_0(\mathfrak{I}) \ar[r]^{\pi_*}& \Z\oplus \Z\ar[d]^{{\mathtt{exp}}} \\
\Z\oplus \Z\ar[u]^{{\mathtt{ind}}} &K_1(\mathfrak{I}) \ar[l]^{\pi_*} &0\ar[l]^{i_*} }
\end{equation}
Therefore, from \ref{seq 3.12} one has 
$$K_0(\mathfrak{I})/{\rm Im}(i_*)\;=\;K_0(\mathfrak{I})/{\rm Ker}(\pi_*)\;\simeq \;{\rm Im}(\pi_*)\;=\;{\rm Ker}({\mathtt{exp}})\;=\;\Z\oplus \Z.$$
Moreover, the relations
\begin{equation*}
    \begin{split}
        [\mathfrak{z}_{0}]+[\mathfrak{r}_{(0,1)}]\;&=\;[\mathfrak{r}_{0}\mathfrak{u}_{0}+\mathfrak{r}_{(0,1)}]\;=\;[\mathfrak{r}_{0}]\;=\;[\mathfrak{r}_{(0,1)}]\\
        [\mathfrak{z}_{0}]+[\mathfrak{u}_{(0,1)}]\;&=\;[\mathfrak{u}_{(1,0)}]
    \end{split}
\end{equation*}
show that $i_*([\mathfrak{z}_{0}])=[\mathfrak{0}]$ and $K_0(\mathfrak{I})=\Z\oplus \Z$. Explicitly,  $K_0(\mathfrak{I})=\Z[\mathfrak{r}_{0}]\oplus\Z[\mathfrak{u}_{0}]$ because $$\pi_*\big([\mathfrak{r}_{0}]\big)\;=\;\big[(\mathfrak{1}\otimes \mathfrak{z}_{0},\mathfrak{0})\big]\;,\hspace{1cm} \pi_*\big([\mathfrak{u}_{0}]\big)\;=\;\big[(\mathfrak{0},\mathfrak{1}\otimes \mathfrak{z}_{0})\big]$$ 
where $\big[(\mathfrak{1}\otimes \mathfrak{z}_{0},\mathfrak{0})\big]$ and $\big[(\mathfrak{0},\mathfrak{1}\otimes \mathfrak{z}_{0})\big]$ are the generators of $K_0\big(\,C(\mathbb{T})\otimes\mathcal{K}\oplus C(\mathbb{T})\otimes\mathcal{K}\big).$
The sequence \ref{seq six} implies that 
\begin{equation}\label{ind}
    K_1(\mathfrak{I})\;\simeq\; {\rm Im}(\pi_*)\;=\;{\rm Ker}({\mathtt{ind}})\;.
\end{equation}
By stability, we have
\begin{equation*}
    \begin{split}
        K_1\big(\,C(\mathbb{T})\otimes\mathcal{K}\oplus C(\mathbb{T})\otimes\mathcal{K}\,\big)\;&=\;K_1\big(\,C(\mathbb{T})\big)\oplus K_1\big( C(\mathbb{T})\,\big)\;.
    \end{split}
\end{equation*}
Thereby we can identified the generators of $ K_1\big(\,C(\mathbb{T})\otimes\mathcal{K}\oplus C(\mathbb{T})\otimes\mathcal{K}\,\big)$
with $[(\mathfrak{s}^\U_1,\mathfrak{1})]$ and $[(\mathfrak{1},\mathfrak{s}^\U_2)]$, where we have considered the isomorphism $C(\mathbb{T})\simeq C^*(\mathfrak{s}_i^\U)$ for $i=1,2.$ By stability again and for the sake of notational simplicity let us remove the “compact part" of the elements of $C(\mathbb{T})\otimes\mathcal{K}\oplus C(\mathbb{T})\otimes\mathcal{K}$ in the following computations. Notice that  $(\mathfrak{s}_1-\mathfrak{1})\mathfrak{r}_0+\mathfrak{1}$ and $(\mathfrak{s}_2-\mathfrak{1})\mathfrak{u}_0+\mathfrak{1}$ are partial isometries lifts of $(\mathfrak{s}^\U_1,\mathfrak{1})$ and $(\mathfrak{1},\mathfrak{s}^\U_2)$, respectively. Indeed, a calculation provides
\begin{equation*}
\begin{split}
     \pi \big( (\mathfrak{s}_1-\mathfrak{1})\mathfrak{r}_0+\mathfrak{1}\big)\;&=\;\big(U {{\rm ev}_\U} \big( (\mathfrak{s}_1-\mathfrak{1})\mathfrak{r}_0+\mathfrak{1}\big) U^{-1},U {{\rm ev}_\U} \big( (\mathfrak{s}_1-\mathfrak{1})\mathfrak{r}_0+\mathfrak{1}\big) U^{-1}\big)\\
     &=\;(\mathfrak{s}^\U_1,\mathfrak{1})-(\mathfrak{0,1})\;=\;(\mathfrak{s}_1^\U,\mathfrak{1})  
\end{split}
\end{equation*}

and
\begin{equation*}
    \begin{split}
      \big( \mathfrak{r}_0(\mathfrak{s}^*_1-\mathfrak{1})+\mathfrak{1}\big)\big(  (\mathfrak{s}_1-\mathfrak{1})\mathfrak{r}_0+\mathfrak{1}\big)&\;=\;\mathfrak{r}_0(\mathfrak{s}^*_1-\mathfrak{1})(\mathfrak{s}_1-\mathfrak{1})\mathfrak{r}_0+\mathfrak{r}_0(\mathfrak{s}^*_1-\mathfrak{1})+(\mathfrak{s}_1-\mathfrak{1})\mathfrak{r}_0+\mathfrak{1}\\
     &\;=\;-\mathfrak{r}_0\mathfrak{s}_1^*\mathfrak{r}_0-\mathfrak{r}_0\mathfrak{s}_1\mathfrak{r}_0+\mathfrak{s}_1\mathfrak{r}_0+\mathfrak{r_0}\mathfrak{s}_1^*+\mathfrak{1}\\
     &\;=\;\mathfrak{1}.
    \end{split}
\end{equation*}
The latter equality is consequence of that $\mathfrak{r}_0\mathfrak{s}_1\mathfrak{r}_0=\mathfrak{s}_1\mathfrak{r}_0$. It follows that
\begin{equation*}
    \begin{split}
    {\mathtt{ind}}\big([(\mathfrak{s}^\U_1,\mathfrak{1})]\big)\;&=\;[\mathfrak{1}-( \mathfrak{r}_0(\mathfrak{s}^*_1-\mathfrak{1})+\mathfrak{1})(  (\mathfrak{s}_1-\mathfrak{1})\mathfrak{r}_0+\mathfrak{1}) ]-[\mathfrak{1}-(  (\mathfrak{s}_1-\mathfrak{1})\mathfrak{r}_0+\mathfrak{1})( \mathfrak{r}_0(\mathfrak{s}^*_1-\mathfrak{1})+\mathfrak{1})]\\
    &=\;[\mathfrak{0}]-[\mathfrak{z}_{(0,0)}]\;=\;-[\mathfrak{z}_{(0,0)}]\;.
    \end{split}
\end{equation*}
The same argument also holds for $(\mathfrak{1,s}^\U_2)$ and in turn
${\mathtt{ind}}\big([(\mathfrak{1},\mathfrak{s}^\U_2)]\big)=-[\mathfrak{z}_{(0,0)}].$ Hence the kernel of ${\mathtt{ind}}$ is equal to $\Z[(\mathfrak{s}^\U_1,{\mathfrak{s}_2^{\U}}^*)]$. Let $\mathfrak{w}$ be the unitary operator in $\mathfrak{I}^+$ given by $\mathfrak{w}:=(\mathfrak{s}_{1}-\mathfrak{1})\mathfrak{r}_{0}+(\mathfrak{s}_2^*-\mathfrak{1})\mathfrak{u}_{0}(\mathfrak{1}-\mathfrak{r}_0)+ \mathfrak{1}$\footnote{A calculation provides that $\mathfrak{w}$ is a unitary element in $\mathfrak{I}^+$, which in fact acts on $\ell^2(\Z^2)$ translating anti-clockwise along the interface}. Since $$\pi(\mathfrak{w})\;=\;\pi(\,(\mathfrak{s}_{1}-\mathfrak{1})\mathfrak{r}_{0}+(\mathfrak{s}_2^*-\mathfrak{1})\mathfrak{u}_{0}(\mathfrak{1}-\mathfrak{r}_0)+ \mathfrak{1}\,)\;=\;(\mathfrak{s}^\U_1,{\mathfrak{s}_2^{\U}}^*)\;,$$
where $\pi$ is given in \ref{10}, then by \ref{ind} one concludes 
$$K_1(\mathfrak{I})\;=\;\Z\big[\,(\mathfrak{s}_{1}-\mathfrak{1})\mathfrak{r}_{0}+(\mathfrak{s}_2^*-\mathfrak{1})\mathfrak{u}_{0}(\mathfrak{1}-\mathfrak{r}_0)+ \mathfrak{1}\big]\;=\;\Z[\mathfrak{w}]\;.$$
\end{proof}

\section{Bulk-interface correspondence}
In this section, we study the quantization of the interface currents and derive bulk-interface correspondence by using tools from \cite{Deni1, Guo} and the K-groups associated with $\A$, $\mathfrak{I}$ and $\A_{{\rm bulk}}.$
\subsection{Bulk topological invariants} For any $\theta=(\theta_1,\theta_2)\in [0,2\pi)^2$ let us consider the unitary operator $\mathfrak{v}_\theta$ which acts on  $\psi\in \ell^2(\Z^2)$  as
\begin{equation}
    (\mathfrak{v}_\theta\psi)(n_1,n_2)\;:=\;\expo{-\ii(\theta_1n_1+\theta_2n_2)}\psi(n_1,n_2)\;.
    \end{equation}
These unitary operators define a continuous group action $[0,2\pi)^2\ni\theta\mapsto \alpha_\theta$ on the bulk algebra given by
\begin{equation}
    \alpha_\theta(\mathfrak{a})\;:=\;\mathfrak{v}_\theta \mathfrak{a}\mathfrak{v}_\theta^*\;,\qquad\mathfrak{a}\in \A_{\rm bulk}\;.
\end{equation}
Thereby, we can introduce the infinitesimal generators $\nabla_1$ and $\nabla_2$ on $\A_{\rm bulk}$ defined as
$$\nabla_1(\mathfrak{a})\;:=\;\lim_{\theta_1\rightarrow0}\frac{\alpha_{(\theta_1,0)}(\mathfrak{a})-\mathfrak{a}}{\theta_1}\;,\qquad \nabla_2(\mathfrak{a})\;:=\;\lim_{\theta_2\rightarrow0}\frac{\alpha_{(0,\theta_2)}(\mathfrak{a})-\mathfrak{a}}{\theta_2}$$
for some suitable elements $\mathfrak{a}\in\A_{\rm bulk}$.
\begin{remark}
 In general $\nabla_1$ and $\nabla_2$ can be introduced for any magnetic algebra \cite[Section 3.5]{Deni1}. Thus without loss of generality, we will use the same notation for these derivations in $\A_{b_\llcorner}$ and $\A_{b_\ast}$.
\end{remark}
Let $\A_{\rm bulk}^0\subset \A_{\rm bulk}$ be the dense subalgebra of non-commutative polynomials in the variables $(\mathfrak{s}_{b_\llcorner,1},\mathfrak{0})$, $(\mathfrak{s}_{b_\llcorner,2}, \mathfrak{0})$, $(\mathfrak{0},\mathfrak{s}_{b_\ast,1})$ and $(\mathfrak{0},\mathfrak{s}_{b_\ast,2})$. Let us introduce the spaces
$$\mathcal{C}^k(\A_{\rm bulk})\;:=\;\overline{\A_{\rm bulk}^0}^{\|\cdot\|_k},$$
where the norm $\|\cdot\|_k$ is given by
$$\|\mathfrak{a}\|_k\;:=\;\sum_{i=0}^k\sum_{a+b=i}\|\nabla_1^a\nabla_2^b \mathfrak{a}\|\;.$$
Since for a constant magnetic field $B$ of strength $b$, there is a unique ergodic measure with the associated trace per unit volume \cite[Proposition 2.28]{Deni1}, then we shall denote as $\mathcal{T}_\llcorner$ and $\mathcal{T}_\ast$ for these unique traces in $\A_{b_\llcorner}$ and $\A_{b_\ast}$, respectively. Given a differentiable projection $\mathfrak{p}=(\mathfrak{p}_\llcorner,\mathfrak{p}_\ast)\in \mathcal{C}^1(\A_{\rm bulk})$, the \emph{transverse Hall conductance} associated with $\mathfrak{p}$ is defined by
\begin{equation}
     \sigma_{{\rm bulk}}(\mathfrak{p})\;:=\;\big(\sigma_{b_{\llcorner}}(\mathfrak{p}),\sigma_{b_\ast}(\mathfrak{p})\big)\;.
\end{equation}
where
\begin{equation*}
\begin{split}
     \sigma_{b_{\llcorner}}(\mathfrak{p})\;:&=\;2\pi \ii \frac{e^2}{h}\mathcal{T}_\llcorner(\mathfrak{p}_\llcorner[\nabla_1\mathfrak{p}_\llcorner,\nabla_2\mathfrak{p}_\llcorner])\\
      \sigma_{b_{\ast}}(\mathfrak{p})\;:&=\;2\pi \ii \frac{e^2}{h}\mathcal{T}_\ast(\mathfrak{p}_\ast[\nabla_1\mathfrak{p}_\ast,\nabla_2\mathfrak{p}_\ast])\;,\\
\end{split}
\end{equation*}
$e>0$ is the magnitude of the electron charge and $h$ is Planck's constant.
\medskip

\noindent
\textbf{Bulk gap condition (BGC).} Let $\mathfrak{\hat{h}}$ be the \emph{full magnetic Hamiltonian}, i.e., a selfadjoint element in $\A$. Let $\mathfrak{h}\equiv(\mathfrak{h}_\llcorner,\mathfrak{h}_\ast):={\rm ev}(\mathfrak{\hat{h}})\in \A_{{\rm bulk}}$ be the \emph{bulk magnetic Hamiltonian} and assume that there is a compact set $\Delta$ such that
$$\min \sigma(\mathfrak{h})\;< \;\min \Delta\;<\;\max \Delta\;<\;\max \sigma(\mathfrak{h})\;.$$

As a consequence of {\bf BGC}, for any $\mu\in\Delta $ the \emph{Fermi projection}
$$\mathfrak{p}_\mu\;:=\;\big(\mathfrak{p}_{\mu_\llcorner},\mathfrak{p}_{\mu_\ast}\big)\;=\;\big(\,\chi_{(-\infty,\mu]}(\mathfrak{h}_\llcorner)\,,\,\chi_{(-\infty,\mu]}(\mathfrak{h}_\ast)\,\big)$$
is an element of bulk algebra. Furthermore, it follows that $\sigma_{b_{\llcorner}}(\mathfrak{p}_\mu)$ and  $ \sigma_{b_{\ast}}(\mathfrak{p}_\mu)$ are quantized. For instance, one has the equality $\sigma_{b_{\llcorner}}(\mathfrak{p}_\mu)=\frac{e^2}{h}{\rm Ch}(\mathfrak{p}_{\mu_\llcorner})$, where ${\rm Ch}(\mathfrak{p}_{\mu_\llcorner})$ is the \emph{Chern number} of the projector $\mathfrak{p}_\llcorner$ and it takes integer values \cite{Con,PRO}. We will call the integers $\sigma_{b_{\llcorner}}(\mathfrak{p}_\mu)$ and $ \sigma_{b_{\ast}}(\mathfrak{p}_\mu)$ as the \emph{bulk magnetic invariants} of the system
\medskip

In view of {\bf BGC}, there is a nondecreasing smooth function $g\colon\R\rightarrow [0,1]$ such that $g=0$ below $\Delta$ and $g=1$ above $\Delta$. Then
\begin{equation*}
    \begin{split}
       {\rm ev}(\mathfrak{1}-g(\mathfrak{\hat{h}}))\;&=\;\mathfrak{1}-g({\rm ev}(\mathfrak{\hat{h}}))\;=\;\mathfrak{1}-(\mathfrak{1}-\mathfrak{p}_\mu)\;=\;\mathfrak{p}_\mu\;.
    \end{split}
\end{equation*}
Observe that the Fermi projector $\mathfrak{p}_\mu$ defines a class $[\mathfrak{p}_\mu]$ in $K_0(\A_{\rm bulk})$ and the unitary operator $\mathfrak{u}_\Delta:=\expo{ 2\pi \ii g(\mathfrak{\hat{h}})}$ defines a class in $K_1(\mathfrak{I})$. From the fact that  $\mathfrak{1}-g(\mathfrak{\hat{h}})$ is a self-adjoint lift of $\mathfrak{p}_\mu$  one gets the following Proposition.
\begin{proposition}\label{prop: 4.1}
Assume {\bf BGC} and let $\mu\in \Delta$. There exist a smooth function $g\colon\R\rightarrow [0,1]$ such that the unitary operator $\mathfrak{u}_\Delta=\expo{ 2\pi \ii g(\mathfrak{\hat{h}})}\in \mathfrak{I}^+$ fulfills
$${\mathtt{exp}}([\mathfrak{p}_\mu])\;=\;-[\mathfrak{u}_\Delta]\;.$$
\end{proposition}
\subsection{Interface currents}
 Recall that the natural trace on $C(\mathbb{T})$ is given by $$\tau_0(f)\;:=\;\int_{\mathbb{T}}f(k)\,{\rm d}k\;,\qquad f\in C(\mathbb{T})\;.$$ Here ${\rm d}k$ is the normalized Haar measure on the one-dimensional torus $\mathbb{T}$. Since $\mathfrak{I}=\A(\mathfrak{r}_0)+\A(\mathfrak{u}_0)$, we obtain a trace $\mathcal{T}_1$ on $\mathfrak{I}$ pulling back the induced trace $\tau_0\otimes {\rm Tr}_{\ell^2(\Z)}$ through the isomorphism 
$$\A(\mathfrak{r}_0)/\A(\mathfrak{u}_0)\cap\A(\mathfrak{r}_0)\;\simeq\; C(\mathbb{T})\otimes\mathcal{K}\;.$$
 The same argument also provides the trace $\mathcal{T}_2$ on $\mathfrak{I}$ by using the isomorphism 
 $$\A(\mathfrak{u}_0)/\A(\mathfrak{u}_0)\cap\A(\mathfrak{r}_0)\;\simeq\; C(\mathbb{T})\otimes\mathcal{K}\;.$$
We will write $\mathcal{D}_{\mathfrak{I},i}\subset \mathfrak{I}$ for the set of trace-class elements of $\mathfrak{I}$ with respect to $\mathcal{T}_i.$ One can endow $\mathfrak{I}$ with the unbounded derivations $\nabla_\mathfrak{I,i}$ given by
$$\nabla_\mathfrak{I,i}\;:=\;[\,\mathfrak{n}_i,\;\cdot\;]\;,\qquad i=1,2$$
where $$\mathfrak{n}_i\psi(n_1,n_2)\;:=\;n_i\psi(n_1,n_2)\;,\qquad\psi\in \ell^2(\Z^2)\;.$$
We can extend such a derivations on the unitization of $\mathfrak{I}$ with the prescription $\nabla_{\mathfrak{I},i}(\mathfrak{1})=0.$
For $k\in \mathbb{N}$ let us introduce the spaces
$$\mathcal{C}^k_{\mathfrak{I},i}\;:=\;\{\,\mathfrak{a}\in\mathfrak{I}\;:\;\nabla_{\mathfrak{I},i}^k(\mathfrak{a})\in \mathcal{D}_{\mathfrak{I},i}\,\}\;,\hspace{1cm}i=1,2 \;.$$
 For any unitary operator $\mathfrak{u}\in \mathfrak{I}^+$ such that $\mathfrak{u}-\mathfrak{1}\in \mathcal{C}_{\mathfrak{I},i}^1$, we define the \emph{winding numbers} of $\mathfrak{u}$ as
\begin{equation}
    W_{\mathfrak{I},i}(\mathfrak{u})\;:=\;\ii\mathcal{T}_i(\mathfrak{u}^*\nabla_{\mathfrak{I},i}(\mathfrak{u}))\;,\qquad i=1,2\;.
\end{equation}

Note that the derivation $\nabla_{\mathfrak{I},i}$ is, in principle, only defined on suitable elements of the interface algebra. In order to make sense to $\nabla_{\mathfrak{I},i}(\hat{\mathfrak{h}})$ for the Hamiltonian $\mathfrak{\hat{h}}$, we introduce the following assumption.

\medskip
\noindent
{\bf Existence of the current operator (ECO).} Assume that the derivations $\nabla_{\mathfrak{I},i}$ can be extended to class of sufficiently regular elements of $\A$ for $i=1,2$. Moreover, we assume that $\nabla_{\mathfrak{I},i}(\hat{\mathfrak{h}})$ exists as element of $\A$  for the full magnetic Hamiltonian $\hat{\mathfrak{h}}.$

There are different ways to perform such an extension. For instance, by adapting the argument used in  \cite{Com, Dom},  one can define $\nabla_{\mathfrak{I},i}=[\chi_i\mathfrak{1}, \,\cdot\,]$ where $\chi_i$ are the \emph{switch functions} given by $\chi_i(n)=n_i$ when $n=(n_1,n_2)$ lies in the interface\footnote{\ie the set $\{(n_1,0)\,:\,n_1\geq 0\}\cup\{(0,n_2)\,:\,n_2\geq 0\} $} and $\chi_i(n)=0$ otherwise. Other example, by combining the derivations $\nabla_1$ and $\nabla_2$ of $\A$, is described in \cite{Guo}.

\begin{remark}
It is important to point out that $\hslash^{-1}\nabla_{\mathfrak{I},i}(\hat{\mathfrak{h}})$ can be physically interpreted as the \emph{ velocity} operator along the faces of the interface. Here $\hslash$ is the reduced Planck's constant.
\end{remark}
Let us suppose {\bf ECO}, then
\begin{equation}
    J_{\mathfrak{I},i}(\Delta)\;:=\;\frac{e}{\hslash}\mathcal{T}_i\big( g'(\mathfrak{\hat{h}})\nabla_{\mathfrak{I},i}(\hat{\mathfrak{h}})\big)\;,\qquad i=1,2
\end{equation}
represents the two current densities along $\mathfrak{I}$. Therefore, by Kubo's Formula \cite{Bel,Deni2}, the terms $\sigma_{\mathfrak{I},i}:=eJ_{\mathfrak{I},i}$ provide  the  \emph{interface conductances}, where recall that $e>0$ is the magnitude of electron charge. Furthermore, if we assume the assumptions ${\bf CGB }$ and ${\bf ECO}$ one obtains
$$\mathcal{T}_i\big(g'(\mathfrak{\hat{h}})\nabla_{\mathfrak{I},i}(\mathfrak{\hat{h}})\big)\;=\;-\frac{1}{2\pi}W_{\mathfrak{I},i}(\mathfrak{u}_\Delta)\;,\qquad i=1,2\;.$$
The latter equality can be obtained by adapting the proof of \cite[Proposition 7.1.2]{PRO}. As a result, the interface conductances associated with the interface states in $\Delta$ are given by
$$\sigma_{\mathfrak{I},i}(\Delta)\;=\;\frac{e^2}{h}W_{\mathfrak{I},i}(\mathfrak{u}_\Delta)\;\qquad i=1,2\;.$$
Notice that in agreement with the above, one can define for general unitary operators $\mathfrak{u}\in \mathfrak{I}^+$ such that $\mathfrak{u}-\mathfrak{1}\in \mathcal{C}_{\mathfrak{I},i}^1$ the interface conductance by
\begin{equation}
    \sigma_{\mathfrak{I},i}(\mathfrak{u})\;:=\;\frac{e^2}{h}  W_{\mathfrak{I},i}(\mathfrak{u})\;.
\end{equation}
\begin{remark}
 The term $\sigma_{\mathfrak{I},1}(\mathfrak{u})$ is the proportionality coefficient of the current flowing along the interface in the $x$-axis when the system is in the configuration $\mathfrak{u}$ \cite{Deni1, PRO}. For $i=2$ it provides the other current in the $y$-axis of the interface.
\end{remark}
\subsection{Proof of the Bulk-interface correspondence} We are in the position to state our main result of this section. Let us begin with some previous technical results.
\begin{lemma}\label{lem: 4.3}
Let $\mathfrak{p}=(\mathfrak{p}_\llcorner,\mathfrak{p}_\ast)$ be a projection in $\A_{{\rm bulk}}.$ Then
$${\mathtt{exp}}([\mathfrak{p}])\;=\;\big({\rm Ch}(\mathfrak{p}_\llcorner)-{\rm Ch}(\mathfrak{p}_\ast)\big)[\mathfrak{w}]\;,$$
where ${\rm Ch}(\mathfrak{p}_\llcorner)$ and ${\rm Ch}(\mathfrak{p}_\ast)$ are the Chern numbers of $\mathfrak{p}_\llcorner$ and $\mathfrak{p}_\ast$, respectively.
\end{lemma}
\begin{proof}
In light of the proposition \ref{prop: proj}, the elements $(\mathfrak{1,0})$ and $(\mathfrak{0,1})$ in $\A_{\rm bulk}$ lift in the projections $\mathfrak{q}_0$ and $\mathfrak{q}^\bot_0$, respectively. Therefore, using  \cite[Proposition 12.2.2]{Ro} one finds 
\begin{equation}\label{expo}
{\mathtt{exp}}([(\mathfrak{1,0})])\;=\;{\mathtt{exp}}([(\mathfrak{0,1})])\;=\;[\mathfrak{1}]\;.
\end{equation}
On the other hand, the class $[\mathfrak{p}]$ can be written as
\begin{equation}\label{expresion pro}
    [\mathfrak{p}]\;=\;M_\llcorner[(\mathfrak{1,0})]+M_\ast[(\mathfrak{0,1})]+{\rm Ch}(\mathfrak{p}_\llcorner)[(\mathfrak{p}_{\theta_\llcorner},\mathfrak{0})]+{\rm Ch}(\mathfrak{p}_\ast)[(\mathfrak{0},\mathfrak{p}_{\theta_\ast})]\;,
\end{equation}
where $\mathfrak{p}_{\theta_\llcorner}$ and $\mathfrak{p}_{\theta_\ast}$ are the Powers-Rieffel projections of $\A_{b_\llcorner}$ and $\A_{b_\ast}$, respectively. We know that 
\begin{equation*}
    \begin{split}
        \mathfrak{p}_{\theta_\llcorner}\;&=\;\mathfrak{s}^*_{b_\llcorner,1}f(\mathfrak{s}_{b_\llcorner,2})+g(\mathfrak{s}_{b_\llcorner,2})+f(\mathfrak{s}_{b_\llcorner,2})\mathfrak{s}_{b_\llcorner,1}\;,\\
        \mathfrak{p}_{\theta_\ast}\;&=\;\mathfrak{s}^*_{b_\ast,1}f(\mathfrak{s}_{b_\ast,2})+g(\mathfrak{s}_{b_\ast,2})+f(\mathfrak{s}_{b_\ast,2})\mathfrak{s}_{b_\ast,1}\;,
    \end{split}
\end{equation*}
for some suitable continuous real functions $f$ and $g$ on $\mathbb{T}$ \cite[Proposition 12,4]{gracia-varilly-figueroa-01}. Consider now the self-adjoint lift of $(\mathfrak{p}_{\theta_\llcorner},\mathfrak{p}_{\theta_\ast})$ given by
\begin{equation*}
    \begin{split}
        \mathfrak{l}\;&=\;\mathfrak{s}^*_{1}f(\mathfrak{s}_{2})+g(\mathfrak{s}_{2})+f(\mathfrak{s}_{2})\mathfrak{s}_{1}\,.
    \end{split}
\end{equation*}
From the construction of $f$ and $g$ it follows that $l$ is a projection in $\A$ and for this reason we have
\begin{equation*}
\begin{split}
    {\mathtt{exp}}\big([(\mathfrak{p}_{\theta_\llcorner},\mathfrak{0})]+[(\mathfrak{0},\mathfrak{p}_{\theta_\ast})]\big)\;&=\;   {\mathtt{exp}}\big([(\mathfrak{p}_{\theta_\llcorner},\mathfrak{p}_{\theta_\ast})]\big)\;=\;[\expo{2\pi \ii \mathfrak{l}}]\;=\;[\mathfrak{1}]
\end{split}
\end{equation*}
Thus, it follows that $ {\mathtt{exp}}\big([(\mathfrak{p}_{\theta_\llcorner},\mathfrak{0})]\big)=-{\mathtt{exp}}\big([(\mathfrak{0},\mathfrak{p}_{\theta_\ast})]\big)$. The latter with \ref{expo} and \ref{expresion pro} imply that
$${\mathtt{exp}}([\mathfrak{p}])\;=\; \big({\rm Ch}(\mathfrak{p}_\llcorner)-{\rm Ch}(\mathfrak{p}_\ast)\big)q[\mathfrak{w}]\;,$$
for some $q\in \Z.$  In order to show that $q=\pm 1$, notice that ${\rm ev}_*\colon K_1(\A)\rightarrow K_1(\A_{\rm bulk})$ in the exact sequence \ref{seq six} is injective and thanks to \ref{seq six} again, ${\mathtt{exp}}$ is a surjective map and so it holds that $q=\pm 1$.  Finally, since $K_0(\A_{\rm bulk})=K_0(\A_{b_\llcorner})\oplus K_0(\A_{b_\ast})$ then it is easy to be convinced that 
$$\pi_{1*}\big({\mathtt{exp}}([(\mathfrak{p}_{\theta_\llcorner},\mathfrak{0})])\big)\;=\;\delta([\mathfrak{p}_{\theta_\llcorner}])\;=\;[(\mathfrak{s}_1^\U,\mathfrak{1})]\;,$$
where $\pi_{1*}\colon K_1(\mathfrak{I})\rightarrow K_1(C(\mathbb{T}))$ is the projection on the first component of the image of $\pi$ in the sequence \ref{10},  and $\delta\colon K_0(\A_{b_\llcorner})\rightarrow K_1(C(\mathbb{T}))$ is the exponential map in \cite[Proposition A]{Pim}. Thus, one concludes that ${\mathtt{exp}}([(\mathfrak{p}_{\theta_\llcorner},\mathfrak{0})])=[\mathfrak{w}]$ and $${\mathtt{exp}}([\mathfrak{p}])\;=\; \big({\rm Ch}(\mathfrak{p}_\llcorner)-{\rm Ch}(\mathfrak{p}_\ast)\big)[\mathfrak{w}]\;.$$
\end{proof}
\begin{lemma}\label{lem: 4.4}
It holds true that
\begin{equation*}
    W_{\mathfrak{I},i}(\mathfrak{w})\;=\;(-1)^{i+1}\;,\qquad i=1,2\;.
\end{equation*}
\end{lemma}
\begin{proof}
We know that

\begin{equation*}
    \begin{split}
        \nabla_{\mathfrak{I},1}(\mathfrak{w}-\mathfrak{1})\;&=\;\big[\mathfrak{n}_1\,,\,(\mathfrak{s}_{1}-\mathfrak{1})\mathfrak{r}_{0}+(\mathfrak{s}_2^*-\mathfrak{1})\mathfrak{u}_{0}(\mathfrak{1}-\mathfrak{r}_0)\big]\;=\;\mathfrak{s}_{1}\mathfrak{r}_{0}\;.
    \end{split}
\end{equation*}
The same computation provides
$ \nabla_{\mathfrak{I},2}(\mathfrak{w}-\mathfrak{1})=-\mathfrak{s}_2^*\mathfrak{u}_{0}(\mathfrak{1}-\mathfrak{r}_0).$
On the other hand, since $\mathcal{T}_i\big(\A(\mathfrak{r}_0)\cap\A(\mathfrak{u}_0)\big)=0$ then one obtains
\begin{equation*}
    \begin{split}
         W_{\mathfrak{I},1}(\mathfrak{w})\;&=\;\mathcal{T}_1\big( (\mathfrak{w}^*-\mathfrak{1})\nabla_{\mathfrak{I},1}(\mathfrak{w}-\mathfrak{1})\big)\\
         &=\;\mathcal{T}_1\Big(\big(\mathfrak{r}_0(\mathfrak{s}^*_{1}-\mathfrak{1})+(\mathfrak{1}-\mathfrak{r}_0)\mathfrak{u}_0(\mathfrak{s}_2-\mathfrak{1}) \big)\big(\mathfrak{s}_{1}\mathfrak{r}_{0}\big)\Big)\\
         &=\;\mathcal{T}_1\big( \mathfrak{r}_0-\mathfrak{r}_0\mathfrak{s}_1\mathfrak{r}_0\big)\;=\; 1\;,
    \end{split}
\end{equation*}
and similarly
\begin{equation*}
    \begin{split}
         W_{\mathfrak{I},2}(\mathfrak{w})\;&=\;\mathcal{T}_2\Big( (\mathfrak{w}^*-\mathfrak{1})\nabla_{\mathfrak{I},2}(\mathfrak{w}-\mathfrak{1})\Big)\;=\;\mathcal{T}_2\Big( -\mathfrak{u}_0+\mathfrak{u}_0\mathfrak{s}_1\mathfrak{u}_0\Big)\;=\; -1\;.
    \end{split}
\end{equation*}
\end{proof}
Now we present the main result of this section.
\begin{theorem}\label{Teo current}
Assume {\bf BGC} and {\bf ECO} holds for the full magnetic Hamiltonian $\hat{\mathfrak{h}}$. Let $\mu\in \Delta$ and assume that the bulk Hamiltonian $\mathfrak{h}$ lies in $\mathcal{C}^k(\A_{\rm bulk})$ for some $k\geq 1$. Then, the interface conductance associated to the unitary operator $\mathfrak{u}_\Delta$ defined in the Proposition \ref{prop: 4.1} can be expressed as difference of the bulk magnetic invariants of the system, \ie 
\begin{equation}
    \sigma_{\mathfrak{I},i}(\mathfrak{u}_\Delta)\;=\;(-1)^{i+1}\big(\sigma_{b_\llcorner}(\mathfrak{p}_\mu)-\sigma_{b_\ast}(\mathfrak{p}_\mu)\big)\;,\hspace{1cm}i=1,2\;.
\end{equation}
\end{theorem}
\begin{proof}
 The Lemmas \ref{lem: 4.4} and \ref{lem: 4.3} together the Proposition \ref{prop: 4.1} yield 
\begin{equation*}
    \begin{split}
         \sigma_{\mathfrak{I},i}(\mathfrak{u}_\Delta)\;&=\;\frac{e^2}{h}  W_{\mathfrak{I},i}(\mathfrak{u}_\Delta)\;=\;\frac{e^2}{h}  W_{\mathfrak{I},i}\big({\mathtt{exp}}(-[\mathfrak{p}_\mu])\big)\\
         &=\;(-1)^{i+1}\frac{e^2}{h}\big({\rm Ch}(\mathfrak{p}_\llcorner)-{\rm Ch}(\mathfrak{p}_\ast)\big)\\
         &=\;(-1)^{i+1}\big(\sigma_{b_\llcorner}(\mathfrak{p}_\mu)-\sigma_{b_\ast}(\mathfrak{p}_\mu)\big)\;.
    \end{split}
\end{equation*}
\end{proof}
The next Proposition states that the Theorem \ref{Teo current} holds even if some geometrical imperfections are introduced into the interface.
\begin{proposition}
    For any $b\in C_0(\Z^2)$ it holds that $\A_{B+b}$ and $\A_B$ are isomorphic as $C^*$-algebras. Consequently, the interface currents persist when there are geometric imperfections that vanish at infinity.
\end{proposition}
\begin{proof}
First of all, let us identify the perturbed flux operator $\mathfrak{f}_{B+b}$ with the function $f_{B+b}(n)=\expo{\ii B(n)}\expo{\ii b(n)}$ on $\Z^2$. Since $C_0(\Z^2)\subset C(\Omega_B)$, then it follows that the function $\Z^2\ni n\mapsto\expo{\ii b(n)}$ lies in $C(\Omega_B)$. Therefore, $\mathfrak{f}_{B+b}\in \mathfrak{F}_B$ and by duality this yields a surjective continuous map $\varphi\colon \Omega_B\rightarrow \Omega_{B+b}$. Moreover, notice that $f_{B+b}$ has the same asymptotic behavior of $f_B$, hence the map $\varphi$ must be injective by \ref{2,14}. Actually, $\varphi$ is a homeomorphism since $\Omega_B$ and $\Omega_{B+b}$  are compact Hausdorff spaces. Thus, we conclude that $\A_{B+b}\simeq \A_B$ as $C^*$-algebras, and the remaining part follows from the fact that the K-theory is an algebraic invariant.
\end{proof}
We finish this section with an example where the Theorem \ref{Teo current} assure the existence of non-trivial interface currents. 
\begin{example}
    Let $b_\llcorner= 2\alpha\pi$ and $b_\ast=2\beta \pi$ so that $\alpha$ and $\beta$ are rational numbers and $\alpha-\beta\notin \Z$. Consider the full magnetic Hamiltonian
    \begin{equation*}
        \mathfrak{\hat{h}}\;:=\;\mathfrak{s}_1+\mathfrak{s}_1^*+\mathfrak{s}_2+\mathfrak{s}_2^*+\mathfrak{v}
    \end{equation*}
    where $\mathfrak{v}$ is a selfadjoint element in the interface algebra $\mathfrak{I}$ such that $\mathfrak{v}$ is in the domain of $\nabla_{\mathfrak{I},i}$ with $i=1,2$. In view of the Proposition \ref{prop 3,2}, the components
    of the corresponding bulk Hamiltonian $\mathfrak{h}=(\mathfrak{h}_\llcorner,\mathfrak{h}_\ast)\equiv {\rm ev}(\hat{\mathfrak{h}})$ fulfills
    \begin{equation*}
        \mathfrak{h}_\llcorner\;=\;\mathfrak{s}_{b_\llcorner,1}+\mathfrak{s}_{b_\llcorner,1}^*+\mathfrak{s}_{b_\llcorner,2}+\mathfrak{s}_{b_\llcorner,2}^*
    \end{equation*}
      \begin{equation*}
        \mathfrak{h}_\ast\;=\;\mathfrak{s}_{b_\ast,1}+\mathfrak{s}_{b_\ast,1}^*+\mathfrak{s}_{b_\ast,2}+\mathfrak{s}_{b_\ast,2}^*
    \end{equation*}
According to \cite[Section 2.1]{Deni1}, in the Landau gauge\footnote{The gauge given by $A(n,n-e_j):=\delta_{j,2}n_1b_\llcorner$ whose circulation is the constant magnetic field $B(n)=b_\llcorner.$}  $\h_\llcorner$ reads
 \begin{equation*}
     \begin{split}
         (\h_\llcorner \psi)(n_1,n_2)\;=&\;\psi(n_1-1,n_2)+\psi(n_1+1,n_2)+\expo{2\pi \alpha n_1}\psi(n_1,n_2-1)\\
         &\;+\expo{-2\pi \alpha n_1}\psi(n_1,n_2+1)
     \end{split}
 \end{equation*}
 for all $\psi\in \ell^2(\Z^2)$. This is a Harper-like operator \cite{Har} and the spectrum of $\h_\llcorner$ is given by the union of $q$ energy bands when $\alpha=p/q$, where $p$ and $q$ relative prime integers \cite[Section 2.6]{Bel}. Moreover, all the energy bands are separated except the central one \cite{Avr,Cho}.  
 Since the same arguments also work for $\h_\ast$, then for suitable values of $\alpha$ and $\beta$ one can choose $\mu$ in a common spectral gap $\Delta$ of $\h_\llcorner$ and $\h_\ast$ so that $\hat{\h}$ meets {\bf BGC} and  ${\rm Ch}(\mathfrak{p}_{\mu_\llcorner})\neq {\rm Ch}(\mathfrak{p}_{\mu_\ast})$.\footnote{The Hofstadter butterflies \cite{AVR1, Avr} provides specific values for $\alpha$ and $\beta$.} \\
 On the other hand, one has that $[\mathfrak{n}_j,\mathfrak{s}_i]=\delta_{i,j}\mathfrak{s}_i\in \A$ and this implies that $\hat{\h}$ meets {\bf ECO}. Finally, Theorem \ref{Teo current}, shows that  $\sigma_{\mathfrak{I},i}(\mathfrak{u}_\Delta)\neq 0.$ 
\end{example}
\section{Corner States}
In this section, we use an adaptation of \cite{Hay} to define corner states associated with the magnetic quarter-plane algebra. We will prove that these corner states have topological properties which depend on the asymptotic behavior of the system.
\subsection{Toeplitz extensions for the magnetic quarter-plane algebra}
Recall that  $\mathtt{U}_\infty\,\cup\,\{\infty_\ast\}\,\cup\,\{\infty_\llcorner\}$ and $\mathtt{R}_\infty\,\cup\,\{\infty_\ast\}\,\cup\,\{\infty_\llcorner\}$ are invariant closed subsets of $\Omega_B$ and thanks to \cite[Proposition 3.11]{Deni1}, there are magnetic algebras $\A_{\U}$ and $\A_{\rz}$ associated to these sets, which are in fact isomorphic to the Iwatsuka magnetic algebra so that the evaluation maps
${\rm ev}_{\U}\colon \A\rightarrow \A_{\U}$ and ${\rm ev}_{\rz}\colon \A\rightarrow \A_{\rz}$ are well defined surjective $*$-homomorphisms. Explicitly, these maps satisfy 
\begin{equation}\label{evaluation 1}
    \begin{split}
        {\rm ev}_\rz(\mathfrak{s}_1)\;&=\;\mathfrak{s}^\rz_1\;,\qquad {\rm ev}_\rz(\mathfrak{s}_2)\;=\;\mathfrak{s}^\rz_2\:,\qquad{\rm ev}_\rz(\mathfrak{f}_B)\;=\;\mathfrak{f}^\rz_I\;,\\
{\rm ev}_\U(\mathfrak{s}_1)\;&=\;\mathfrak{s}^\U_1\;,\qquad {\rm ev}_\U(\mathfrak{s}_2)\;=\;\mathfrak{s}^\U_2\;,\qquad{\rm ev}_\U(\mathfrak{f}_B)\;=\;\mathfrak{f}^\U
_I\;.
    \end{split}
\end{equation}
\begin{proposition}\label{prop ker}
It holds true that ${\rm Ker}({\rm ev}_\rz)=\A(\mathfrak{u}_0)$ and ${\rm Ker}({\rm ev}_\U)=\A(\mathfrak{r}_0)$.
\end{proposition}
\begin{proof}
First of all, note that $\tau_{(n_1,0)}(\mathfrak{f}_I^\rz)=\mathfrak{f}_I^\rz$ for every $n_1\in \Z$. The latter with the expression of $\mathfrak{u}_0$ given in the proof of Proposition \ref{Prop projections} imply that  $\A(\mathfrak{u}_0)\subset {\rm Ker}({\rm ev}_\rz).$  On the other hand, one knows the decomposition $$\tau_n(\mathfrak{f}_B)\;=\;\big(\tau_{(n_1,n_2)}(\mathfrak{f}_B)-\tau_{(0,n_2)}(\mathfrak{f}_B)\big)+\tau_{(0,n_2)}(\mathfrak{f}_B)\;,\qquad n=(n_1,n_2)\in \Z^2$$  
then 
$${\rm ev}_\rz\big(\tau_n(\mathfrak{f}_B)\big)\;=\;{\rm ev}_\rz\big(\tau_{(0,n_2)}(\mathfrak{f}_B)\big)\;=\;\tau_{(0,n_2)}(\mathfrak{f}_I^\rz)\;.$$
Therefore, if $\mathfrak{g}\in \mathfrak{F}_B$ and ${\rm ev}_\rz(\mathfrak{g})=0$ 
 in turns $\mathfrak{g}\in \A(\mathfrak{u}_0)$. Thereby, the reverse inclusion holds by \ref{interface}. Finally, notice also that $\tau_{(0,n_2)}(\mathfrak{f}_I^\U)=\mathfrak{f}_I^\U$ and hence the same argument shows that ${\rm Ker}({\rm ev}_\U)=\A(\mathfrak{r}_0)$.
 \end{proof}

Using again \cite[Proposition 3.11]{Deni1}, there are two evaluations homomorphisms ${\rm ev}_{\rz,b}\colon \A_{\rz}\rightarrow \A_{b_\llcorner}\oplus\A_{b_\ast}$ and ${\rm ev}_{\U, b}\colon \A_{\U}\rightarrow \A_{b_\llcorner}\oplus\A_{b_\ast},$ which satisfy that
\begin{equation*}
    \begin{split}
        &{\rm ev}_{\rz,b}(\mathfrak{s}^\rz_1)\;=\;\big(\,\mathfrak{s}_{b_\llcorner,1},\,\mathfrak{s}_{b_\ast,1}\,\big)\;,\qquad\hspace{0,3cm} {\rm ev}_{\U,b}(\mathfrak{s}^\U_1)\;=\;\big(\,\mathfrak{s}_{b_\llcorner,1},\,\mathfrak{s}_{b_\ast,1}\,\big)\;,\\
       & {\rm ev}_{\rz,b}(\mathfrak{s}^\rz_2)\;=\;\big(\,\mathfrak{s}_{b_\llcorner,2},\,\mathfrak{s}_{b_\ast,2}\,\big)\;,\qquad\hspace{0,3cm}{\rm ev}_{\U,b}(\mathfrak{s}^\U_2)\;=\;\big(\,\mathfrak{s}_{b_\llcorner,2},\,\mathfrak{s}_{b_\ast,2}\,\big)\;,\\
       & {\rm ev}_{\rz,b}(\mathfrak{f}^\rz_I)\;=\;\big(\,\expo{\ii b_{\llcorner}}\mathfrak{1},\,\expo{\ii b_\ast}\mathfrak{1}\big)\;,
       \qquad{\rm ev}_{\U,b}(\mathfrak{f}^\U_I)\;=\;\big(\,\expo{\ii b_{\llcorner}}\mathfrak{1},\,\expo{\ii b_\ast}\mathfrak{1}\big)\;.
    \end{split}
\end{equation*}
Let us consider the \emph{asymptotic} algebra $\A_{\U,\rz}$, which is the pullback of the two latter $*$-homomorphisms. Namely,
$\A_{\U,\rz}\;:=\;\{ (\mathfrak{u},\mathfrak{r})\in \A_\U\oplus\A_\rz\;|\; {\rm ev}_{\U,b}(\mathfrak{u})={\rm ev}_{\rz,b}(\mathfrak{r})\}$. Observe also that by definition of pullback, the diagram 
\begin{equation}\label{pullback}
    \xymatrix{
&\A_{\U,\rz}\ar[d]^{\pi_\U} \ar[r]^{\pi_\rz} & \A_\rz \ar[d]^{{\rm ev}_{\rz,b}} \\ & \A_\U\ar[r]_{{\rm ev}_{\U,b}}&\A_{b_\llcorner}\oplus\A_{b_\ast}  }
\end{equation}
is commutative, where $\pi_\U$ and $\pi_\rz$ are the projections on the first and second coordinate of $\A_{\U,\rz}$, respectively.
\begin{remark}\label{comen}
It is important to point out that the following diagram is commutative
\begin{equation}
    \xymatrix{
\A \ar[dr]^{{\rm ev}} \ar[r]^{{\rm ev}_\rz}  \ar[d]_{{\rm ev}_\U}& \A_\rz \ar[d]^{{\rm ev}_{\rz,b}}\\
\A_\U \ar[r]_{\!\!\!\!\!\!\!\!\!{\rm ev}_{\U,b}}& \A_{b_\llcorner}\oplus\A_{b_\ast}}
\end{equation}
where ${\rm ev}\colon \A\rightarrow \A_{b_\llcorner}\oplus\A_{b_\ast}$ is the evaluation map given in \ref{sec 3.1}.
\end{remark}
  From the universal property of the pullback $C^*$-algebra and remark \ref{comen}, there is a $*$-homomorphism $\gamma\colon\A\rightarrow \A_{\U,\rz}$ so that the diagram
\begin{equation}\label{pull seq}
    \xymatrix{
\A\ar@/_/[rdd]_{{\rm ev}_\U} \ar@/^/[rrd]^{{\rm ev}_\rz} \ar[rd]^{\gamma} & &\\ &\A_{\U,\rz}\ar[d]^{\pi_\U} \ar[r]^{\pi_\rz} & \A_\rz \ar[d]^{{\rm ev}_{\rz,b}} \\ & \A_\U\ar[r]_{\!\!\!{\rm ev}_{\U,b}}&\A_{b_\llcorner}\oplus\A_{b_\ast}  }
\end{equation}
is commutative. Explicitly, it holds that $\gamma(\mathfrak{a})=\big({\rm ev}_\U(\mathfrak{a}),{\rm ev}_\rz(\mathfrak{a})\big)$ for any $\mathfrak{a}\in \A.$

\begin{theorem}\label{Teo seq}
The sequence 
\begin{equation}
     \xymatrix{
 0\ar[r]&\mathcal{K}\ar[r]^{\hspace{0,1cm}i} & \A \ar[r]^{\gamma}& \A_{\U,\rz}\ar[r]&0}
\end{equation}
is exact, where $\mathcal{K}\simeq \A(\mathfrak{r}_0)\cap \A(\mathfrak{u}_0)$ is the set of compact operators on $\ell^2(\Z^2).$
\end{theorem}
\begin{proof}

Let $(\mathfrak{x',y'})\in \A_{\U,\rz}$ and consider two arbitrary elements $\mathfrak{x}, \mathfrak{y}\in \A$ such that  ${\rm ev}_\U(\mathfrak{x})=\mathfrak{x}'$ and ${\rm ev}_\rz(\mathfrak{y})=\mathfrak{y}'.$ Since by definition ${\rm ev}_{\U,b}(\mathfrak{x'})={\rm ev}_{\rz,b}(\mathfrak{y'})$, then from the diagram \ref{pull seq} one has
\begin{equation*}
    \begin{split}
        {\rm ev}_{\U,b}\big({\rm ev}_\U(\mathfrak{x}-\mathfrak{y})\big)\;&=\;{\rm ev}_{\U,b}(\mathfrak{x'})-{\rm ev}_{\U,b}\big({\rm ev}_\U(\mathfrak{y})\big)\;=\;{\rm ev}_{\U,b}(\mathfrak{x}')-{\rm ev}_{\rz,b}\big({\rm ev}_\rz(\mathfrak{y})\big)\\
        &=\;{\rm ev}_{\U,b}(\mathfrak{x}')-{\rm ev}_{\rz,b}(\mathfrak{y}')\;=\;(\mathfrak{0},\mathfrak{0}).
    \end{split}
\end{equation*}
Therefore, ${\rm ev}_\U(\mathfrak{x}-\mathfrak{y})\in {\rm Ker}({\rm ev}_{\U,b})= \A_\U({\rm ev}_\U(\mathfrak{u}_0))$, where $\A_\U({\rm ev}_\U(\mathfrak{u}_0))$ denotes the two-sided ideal generated by ${\rm ev}_\U(\mathfrak{u}_0)$ in $\A_\U$. The same argument provides that ${\rm ev}_\rz(\mathfrak{x}-\mathfrak{y})\in {\rm Ker}({\rm ev}_{\rz,b})= \A_\rz({\rm ev}_\rz(\mathfrak{r}_0))$.\footnote{The equalities ${\rm Ker}({\rm ev}_{\U,b})=\A_{\U}({\rm ev}_{\U}(\mathfrak{u}_0))$ and ${\rm Ker}({\rm ev}_{\rz,b})=\A_{\rz}({\rm ev}_{\rz}(\mathfrak{r}_0))$ follow from \cite[Proposition 4.6]{Deni1}.}
In light to the Propositions \ref{prop 3,2} and \ref{prop ker}, for some $\mathfrak{r}'\in \A(\mathfrak{r}_0)$ and $\mathfrak{u}'\in\A(\mathfrak{u}_0)$ one has
$$\mathfrak{x}-\mathfrak{y}+\mathfrak{r}'\in\A(\mathfrak{u}_0)\;,\qquad \mathfrak{x}-\mathfrak{y}+\mathfrak{u}'\in\A(\mathfrak{r}_0)\;. $$
Hence there is $\mathfrak{a}\in \A(\mathfrak{r}_0)\cap \A(\mathfrak{u}_0)$ such that $$\mathfrak{x}+\mathfrak{r}'=\mathfrak{y}-\mathfrak{u}'+\mathfrak{a}\;.$$
Choosing $\mathfrak{z}=\mathfrak{x}+\mathfrak{r}'$ one obtains that
\begin{equation*}
    \begin{split}
        \gamma(\mathfrak{z})\;&=\;\big(\,{\rm ev}_\U(\mathfrak{z}),\,{\rm ev}_\rz(\mathfrak{z})\,\big)\;=\;\big(\,{\rm ev}_\U(\mathfrak{x}+\mathfrak{r}'),\,{\rm ev}_\rz(\mathfrak{y}-\mathfrak{u}'+\mathfrak{a})\,\big)\\
        &=\;\big(\,{\rm ev}_\U(\mathfrak{x}),\,{\rm ev}_\rz(\mathfrak{y})\,\big)=(\mathfrak{x',y'})\;.
    \end{split}
\end{equation*}
This verifies that $\gamma$ is a surjective $*$-homomorphism. As a consequence of that $${\rm Ker}(\gamma)\;=\;{\rm Ker}({\rm ev}_\U)\cap {\rm Ker}({\rm ev}_\rz)\;=\; \A(\mathfrak{u}_0)\cap \A(\mathfrak{r}_0)\;\simeq\; \mathcal{K}\;,$$
it follows that ${\rm Ker}(\gamma)={\rm Im}(i)$ and this concludes the proof.
\end{proof}
As an immediate consequence of Theorem \ref{Teo seq}, one can deduce that the asymptotic algebra contains the information of the system that is far from the corner.
Moreover, by Atkinson's Theorem, we obtain the following Proposition.
\begin{proposition}\label{prop: fred}
An element $\mathfrak{t}\in \A$ is Fredholm if and only if ${\rm ev}_\rz(\mathfrak{t})$ and ${\rm ev}_\U(\mathfrak{t})$ are invertible elements in $\A_\rz$ and $\A_\U$, respectively.
\end{proposition}
\medskip
\begin{corollary}\label{Coro seq}
It holds true that $K_0(\A_{\U,\rz})\simeq K_0(\A)$ and $K_1(\A_{\U,\rz})\simeq K_1(\A)\oplus K_0(\K)$ as abelian groups.
\end{corollary}
\begin{proof}
The Theorem \ref{Teo seq} implies the following exact sequence
\begin{equation}\label{exac fin}
    \xymatrix{
K_0(\mathcal{K})= \Z[\mathfrak{z}_0]\ar[r]^{\;\;\;\;\;\;i_*} & K_0(\A) \ar[r]^{\gamma_*\;\;\;}& K_0(\A_{\U,\rz})\ar[d]^{{\mathtt{exp}}} \\
K_1(\A_{\U,\rz})\ar[u]^{{\mathtt{ind}}} &K_1(\A)
\ar[l]^{\;\;\;\;\gamma_*} &0 \ar[l]^{i_*} }
\end{equation} 
Notice that $i_*([\mathfrak{z}_0])=[\mathfrak{z}_0]=[\mathfrak{r}_0-\mathfrak{s}_1\mathfrak{r}_0\mathfrak{s}_1^*]=[\mathfrak{r}_0]-[\mathfrak{s}_1\mathfrak{r}_0\mathfrak{s}_1^*]=0$, and consequently $i_*\colon \Z[\mathfrak{z}_0]\to K_0( \A)$ is the zero map. Therefore, $\gamma_*\colon  K_0(\A)\to  K_0(\A_{\U,\rz})$ is an isomorphism of groups and as a result $K_0(\A_{\U,\rz})\simeq K_0(\A)=\Z^3$. Furthermore, the same argument provides $K_1(\A_{\U,\rz})\simeq K_1(\A)\oplus K_0(\mathcal{K}) = \Z^3$.

\end{proof}

\subsection{Corner Invariants}
Thanks to the Theorem \ref{Teo seq}, we have the short exact sequence
 \begin{equation}\label{seq corner}
     \xymatrix{
 0\ar[r]&\mathcal{K}\otimes C(\mathbb{T})\ar[r]^{} & \A\otimes C(\mathbb{T}) \ar[r]^{}& \A_{\U,\rz}\otimes C(\mathbb{T})\ar[r]&0,}
\end{equation}
\noindent
Let $\hat{\mathfrak{h}}\colon \mathbb{T}\rightarrow\A$ be the time-dependent magnetic Hamiltonian, that is, a self-adjoint element of $\A\otimes C(\mathbb{T})$ so that $\hat{\mathfrak{h}}(0)$ is the full magnetic Hamiltonian $\hat{\mathfrak{h}}$.  We will also consider the \emph{asymptotic Hamiltonian} defined by $\mathfrak{h_a}(t)\equiv(\mathfrak{h}^\U(t),\mathfrak{h}^\rz(t)):=\gamma(\hat{\mathfrak{h}}(t))$.
\medskip

Let $\mathscr{F}^{sa}_*\subset \mathcal{B}(\ell^2(\Z^2))$ be the set of all self-adjoint Fredholm operators such that its essential spectrum is not contained in either $(-\infty,0)$ or $(0,+\infty)$. Let us consider the subspace of $\mathscr{F}^{sa}_*$ given by
$$\mathscr{F}^\infty_*\;:=\;\{ a\in \mathscr{F}^{sa}_*\,|\, \|a\|=1,\;\sigma(a)\mbox{ is finite, and}\;\sigma_{{\rm ess}}(a)=\{\pm 1\}\,\}\;.$$
It is true that the inclusion $i\colon\mathscr{F}^\infty_*\rightarrow\mathscr{F}^{sa}_*$ is an homotopy equivalence \cite{Ati, Phi}.
\medskip

\noindent
\textbf{Asymptotic gap condition (AGC):} Assume that  the Fermi level $\mu$ is not contained in either $\sigma\big(\mathfrak{h}^\U(t)\big)$ or $\sigma\big(\mathfrak{h}^\rz(t)\big)$ for any $t$ in $\mathbb{T}$. We further assume that the spectrum of $\mathfrak{h_a}(t)$ is not contained in either $\R_{<0}$ or $\R_{>0}.$
\medskip

Without loss of generality, in what follows let us assume that the Fermi level $\mu$ is equal to $0.$ From {\bf AGC}, there is a smooth function $f\colon\R\rightarrow[0,1]$ such that $f(\mathfrak{h_a}(t))=\chi_{(-\infty,0]}(\mathfrak{h_a}(t)),$ which is in fact a time-dependent projection in $\A_{\U,\rz}$ and we will denote it by $\mathfrak{p_a}.$ Note also that $\mathfrak{p_a}$ defines an element $[\mathfrak{p_a}]$ in the K-group $K_0\big(\A_{\U,\rz}\otimes C(\mathbb{T})\big)$ and we shall call this class as the \emph{magnetic asymptotic invariant} of the system.
\medskip 

Using {\bf AGC} again and Proposition \ref{prop: fred}, we obtain that the continuous family $\mathfrak{\hat{h}}(t)$ lies in $\mathscr{F}_*^{sa}$ for any $t\in \mathbb{T}$. Thereby the class $[\mathfrak{\hat{h}}]$ in $[\mathbb{T},\,\mathscr{F}^{sa}_*]$ defines an element in 
 $K_1\big(C(\mathbb{T})\big)$ under the identification
 $K_1\big(C(\mathbb{T})\big)\simeq [\mathbb{T},\,\mathscr{F}^{sa}_*]$ given by the \emph{spectral flow}\footnote{The net numbers of eigenvalues crossing $0$ (counting multiplicity) in the positive direction as $t$ goes from $0$ to $1$ \cite{Phi}}   of the family $(\mathfrak{\hat{h}}(t))_{t\in \T}$ (see \cite{Ati, Phi}). We will call the class $[\mathfrak{\hat{h}}]$ as the \emph{magnetic corner invariant} of the system.  
\subsection{Proof of the correspondence} Let us consider the map $\delta_0\colon K_0(\A_{\U,\rz}\otimes C(\T))\rightarrow K_1(C(\mathbb{T}))$ defined as
 
 \begin{equation*}
    \xymatrix{
K_0(\A_{\U,\rz}\otimes C(\mathbb{T}))\ar[rd]^{\delta_0} \ar[r]^{\!{\mathtt{exp}}} &  K_1(\K\otimes C(\mathbb{T}))\ar[d] \\ & K_1(C(\mathbb{T})) }
\end{equation*}
 where the stabilization gives the vertical arrow and the exponential map comes from the sequence \ref{seq corner}. 
Observe that $\delta_0$ is a surjective homomorphism, since in view of the \cite[Exercise 8.B]{Ols},  the six-term exact sequence associated to \ref{seq corner} is equivalent to
\begin{equation*}
    \xymatrix{
K_0(\K)\oplus K_1(\K)\ar[r]^{i_*\oplus i_*} & K_0(\A)\oplus K_1(\A) \ar[r]^{\gamma_*\oplus \gamma_*\;\;\;}& K_0(\A_{\U,\rz})\oplus K_1(\A_{\U,\rz})\ar[d]^{\mathtt{exp}\oplus{\mathtt{ind}}} \\
K_1(\A_{\U,\rz})\oplus K_0(\A_{\U,\rz})\ar[u]^{{\mathtt{ind}\oplus \mathtt{exp}}} &K_1(\A)\oplus K_0(\A)
\ar[l]^{\;\;\;\;\gamma_*\oplus\gamma_*} & K_1(\K)\oplus K_0(\K)\ar[l]^{i_*\oplus i_*} }
\end{equation*} 
Now by Corollary \ref{Coro seq},  one knows that
\begin{equation*}
    K_0(\A_{\U,\rz})\oplus K_1(\A_{\U,\rz})\;=\;K_0(\A)\oplus K_1(\A)\oplus K_0(\K)
\end{equation*}
Thus, for some $N\in \N$ choosing the unitary $u\in {\rm Mat}_N(\A_{\U,\rz})$ such that $\mathtt{ind}([u])=[\mathfrak{z}_0]$  in the sequence  \ref{exac fin} it follows from its very definition that
$\delta_0(\theta([u]))$ is the $\Z$-generator of $K_1(C(\T))$, where $\theta\colon K_1(\A_{\U,\rz})\to K_0\big(C_0((0,1))\otimes \A_{\U,\rz}\big)$ is the map given in \cite[Theorem 7.2.5]{Ols}.
\medskip

 The next Theorem concerns the correspondence between the corner and asymptotic magnetic invariants defined in the quarter-plane system under the map $\delta_0$, which is the main result of this section.
\begin{theorem}\label{teo: corner}
If {\bf AGC} holds, then it is true that $\delta_0([\mathfrak{p_a}])=[\mathfrak{\hat{h}}].$
\end{theorem}
\begin{proof}
Let $\mathfrak{l}$ be a selfadjoint lift of $\mathfrak{q_a}:=\mathfrak{1-p_a}$. In light of the Proposition \ref{prop ker} and Theorem \ref{Teo seq}, it holds that $\sigma_{\rm ess}(\mathfrak{l}(t))=\sigma(\mathfrak{q_a}(t))=\{0,1\}$. Consequently, for $0<\epsilon<1$ one has that $\mathfrak{l}(t)-\epsilon\mathfrak{1}\in \mathscr{F}_*^{sa}$ for any $t$, so let $\mathfrak{\hat{l}}(t)\in \mathscr{F}^\infty_*$ be the element given by the image of $\mathfrak{l}(t)-\epsilon \mathfrak{1}$ under the homotopy equivalence $\mathscr{F}^{sa}_*\simeq_h\mathscr{F}_*^\infty.$ It follows that
\begin{equation*}
    \begin{split}
        {\mathtt{exp}}([\mathfrak{p_a}])\;&=\;[\expo{2\pi\ii\mathfrak{l}} ]\;=\;[\expo{2\pi\ii \mathfrak{\hat{l}}}]\;=\;[\mathfrak{2\hat{l}-1}]\;,
    \end{split}
\end{equation*}
where the last step comes from the homotopy equivalence given in \cite[Proposition 5]{Phi}. On the other hand, observe that for all $t\in \mathbb{T}$  the elements $\chi_{(-\infty,0]}(\mathfrak{\hat{h}}(t))\mathfrak{\hat{h}}(t)-\mathfrak{1}$ and $\mathfrak{l}(t)-\mathfrak{2}$ have essential spectrum contained in $(-\infty,0)$, and the elements $\big(\mathfrak{1}-\chi_{(-\infty,0]}(\mathfrak{\hat{h}}(t))\big)\mathfrak{\hat{h}}(t)+\mathfrak{1}$ and $\mathfrak{l}(t)+\mathfrak{1}$
 have essential spectrum contained in $(0,+\infty)$. Since $\mathfrak{l}$ is a lift of $\mathfrak{q_a}$, then  \cite[Proposition 1]{Phi} and Theorem \ref{Teo seq} imply that $\mathfrak{\hat{h}}\simeq_h \mathfrak{2l-1}\simeq_h\mathfrak{2\hat{l}-1}$, where recall that the symbol $\simeq_h$ means homotopy equivalence of loops.  Thus, one concludes that  $\delta_0([\mathfrak{p_a}])=[\mathfrak{\hat{h}}].$
\end{proof}
The following Corollary provides a criterion for the triviality of the corner states. The proof is an immediate consequence of the Theorem \ref{teo: corner}.

\begin{corollary}\label{coro fin}
A necessary condition for non-trivial corner states is that the time-dependent  magnetic Hamiltonian is gapless.
\end{corollary}
\begin{remark}
When $b_\llcorner=0$ and $b_\ast\notin 2\pi \Z$, one can consider  time-dependent Hamiltonians of the form 
    $\mathfrak{\hat{h}}(t)=\mathfrak{q}_0^\bot+\mathfrak{q}_0\mathfrak{h}(t)\mathfrak{q}_0$
where $\mathfrak{h}(t)$ is a selfadjoint element in $\A$. Since $\mathfrak{q}_0\mathfrak{h}(t)\mathfrak{q}_0$ can be regarded as a operator acting on  $\ell^2(\N\times \N)$ and moreover ${\rm sf}\big((\mathfrak{\hat{h}}(t))_{t\in \T}\,|\,\mu=0\big)\,=\,{\rm sf}\big((\mathfrak{q}_0\mathfrak{h}(t)\mathfrak{q}_0)_{t\in \T}\,|\,\mu=0\big)$\footnote{Along this work the notation ${\rm sf}\big((\mathfrak{h}(t))_{t\in \T}\,|\,\mu=0\big)$ stands for the spectral flow at $\mu=0$ of the family $(\mathfrak{h}(t))_{t\in \T}$}, then the constructions introduced in \cite[Section 4]{Hay}  provides examples of corner states. However, in order to adapt the constructions of \cite{Hay} the Hamiltonians must belong to $ {\rm Mat}_2(\A)$, but by stability of the K-theory,  the Theorem \ref{teo: corner} also applies for these cases.
\end{remark}
Let us describe two examples where the Theorem \ref{teo: corner} applies.
\begin{example} Consider the selfadjoint element $\hat{\mathfrak{h}}:=\mathfrak{r}_0-\mathfrak{r}_0^\bot$ in $\A$. First of all, notice that $\sigma(\hat{\mathfrak{h}})=\{-1,1\}$ and moreover
$$\mathfrak{s}_1\hat{\mathfrak{h}}\mathfrak{s}_1^*\;=\;\mathfrak{r}_{(1,0)}-\mathfrak{z}_0-\mathfrak{r}_0^\bot\;.$$
    Then for $t\in [0,1]$, define
\begin{equation*}
\mathfrak{\hat{h}}(t)\;=\;\mathfrak{u}^*(t)\begin{pmatrix}
\mathfrak{r}_{(1,0)}+(\mathfrak{1}-\mathfrak{2}t)\mathfrak{z}_0-\mathfrak{r}_0^\bot& \mathfrak{0}\\
\mathfrak{0}& \mathfrak{1}
\end{pmatrix}\mathfrak{u}(t)\in  {\rm Mat}_2(\A)
\end{equation*}
where $\mathfrak{u}(t)$ is a continuous unitary path so that
\begin{equation*}
    \mathfrak{u}(0)\;=\;\begin{pmatrix}\mathfrak{1}& \mathfrak{0}\\
    \mathfrak{0}&\mathfrak{1}\end{pmatrix}\;,\qquad \mathfrak{u}(1)\;=\;\begin{pmatrix}\mathfrak{s}_1& \mathfrak{0}\\
    \mathfrak{0}&\mathfrak{s}_1^*\end{pmatrix}
\end{equation*}
It turns out that $\mathfrak{\hat{h}}(0)=\mathfrak{\hat{h}}(1)$, that is, $\mathfrak{\hat{h}}(\cdot)\in {\rm Mat}_2(\A)\otimes C(\T)$. Furthermore, for all $t$ one has $\sigma(\mathfrak{\hat{h}}(t))=\{-1, t,1\}$ and $\sigma_{\rm ess}(\mathfrak{\hat{h}}(t))=\{-1,1\}$. As a result, this Hamiltonian meets ${\bf AGC}$ and hence Theorem \ref{teo: corner} guarantee non-trivial topological corner states since one can see that ${\rm sf}\big(\mathfrak{\hat{h}}(t))_{t\in [0,1]}\;|\;\mu=0\big)= -1.$
\end{example}
\begin{example} For piezoelectric materials, if we add a suitable periodic perturbation along the positive side of the $y$-axis\footnote{As a consequence of that $\mathfrak{u}_n$ and $\mathfrak{r}_n$ are in $\A$ for all $n\in \Z^2$, we can choose any other semi-axis given by the range of these projectors.} the corner states are in correspondence with the  \emph{polarization} in the edge of such a semi-axis. In order to see the latter, let us first say a few words about the \emph{piezoelectric effect}.\\
It is well known that deformations in piezoelectric materials lead to accumulation of charge in the edge of the sample, which are due to two contributions: the relative displacements of the ionic cores,  and electrical conduction which is the so-called \emph{orbital polarization}. Here we will be dealing only with the latter contribution in dimension $1$ for a discrete non-random system where the interaction between particles are neglected. In this case the Hilbert space and the observable algebra turn out to be $\ell^2(\Z)$ and  $C^*(S)\simeq C(\T)$, respectively, where $S$ stands for the usual shift operator on $\ell^2(\Z)$ (\cf \cite{Dani,PRO}). At fixed Fermi level $\mu=0$, the periodic deformation of the system is modelled by a differentiable path of selfadjoint elements $[0,1]\ni t\mapsto H(t)\in C^*(S)$ so that $H(0)=H(1)$, $0\notin \sigma(H(t))$ for all $t$, and there are states above and below of $0$ during the deformation, \ie, the instantaneous Fermi projection $P(t)=\chi_{(-\infty,0]}(H(t))$ is different to the zero or identity operator for all $t$. Then one knows from \cite{Dani, KSV, Pan, Shulz} that, up to arbitrarily small corrections in the adiabatic limit, the orbital polarization $\Delta \mathscr{P}$ accumulated during one adiabatic cycle is given by
\begin{equation}\label{KSV2}
    \Delta\mathscr{P}\;=\;\ii\int_0^1 {\rm d}t\,\mathscr{T}\big( P(t) [\, \partial_t P(t), \,\nabla P(t) \,]\big)
\end{equation}
where $\mathscr{T}$ denotes the trace per unit volume on $C^*(S)$ and $\nabla=\ii [N,\;\cdot\;]$ is the commutator with the position operator. It is important to point out that equation \ref{KSV2} states that, for periodic deformation, the orbital polarization is a bulk effect the topological nature taking values in $2\pi \Z$ \cite[Corollary 1]{Dani}, as noted by Thouless \cite{Thou} in a more restricted context. Now, let us consider the path $[0,1]\ni t\mapsto \hat{H}(t)\in C^*(\hat{S})$ given by the truncation of  $H(t)$ on the Hilbert space $\ell^2(\N)$. From the bulk-boundary correspondence \cite[Theorem 5.5.3 and Section 7.7]{PRO}, it follows that
\begin{equation}
    \Delta \mathscr{P}\;=\;-2\pi\,{\rm sf}\big(\hat{H}(t))_{t\in [0,1]}\;|\;\mu=0\big)
\end{equation}
Now let us connect the latter result with the corner states. For that, we know that there exists a continuous path $[0,1]\in t\mapsto \mathfrak{\hat{h}}_1(t)\in \A$ of selfadjoint elements such that $\mathfrak{u}_0\mathfrak{\hat{h}}_1(t)\mathfrak{u}_0=\hat{H}(t)$ as operators acting on $\ell^2(\N).$ The above is a consequence of that $\mathfrak{u}_0\mathfrak{s}_j\mathfrak{u}_0=\delta_{2,j}\hat{S}$ acting on $\ell^2(\N)$ for $j=1,2$. Here $\hat{S}$ denotes the truncation of the shift operator $S$ on $\ell^2(\N).$ Define
\begin{equation*}
    \mathfrak{\hat{h}}(t)\;:=\;\mathfrak{u}_0\mathfrak{\hat{h}}_1(t)\mathfrak{u}_0+\mathfrak{u}_0^\bot\mathfrak{\hat{h}}_2\mathfrak{u}^\bot_0\in \A\;,\qquad t\in [0,1]
\end{equation*}
for some suitable selfadjoint element $\mathfrak{\hat{h}}_2\in \A$ such that $\mathfrak{u}_0^\bot\mathfrak{\hat{h}}_2\mathfrak{u}^\bot_0$ is invertible as operator acting on  $\ell^2\big(\Z^2\setminus (\{0\}\times \N)\big)$. By using the decomposition $\ell^2(\Z^2)=\ell^2(\N)\oplus \ell^2\big(\Z^2\setminus (\{0\}\times \N)\big)$ one gets for all $t$ $$\sigma(\hat{\mathfrak{h}}(t))\;=\;\sigma(\mathfrak{u}_0\mathfrak{\hat{h}}_1(t)\mathfrak{u}_0)\cup \sigma (\mathfrak{u}_0^\bot\mathfrak{\hat{h}}_2\mathfrak{u}^\bot_0)\;=\;\sigma(\hat{H}(t))\cup \sigma (\mathfrak{u}_0^\bot\mathfrak{\hat{h}}_2\mathfrak{u}^\bot_0).$$
Therefore, the gap assumption on the deformation verifies that
$\mathfrak{\hat{h}}(t)$ meets {\bf AGC}. Furthermore,
\begin{equation*}
\begin{split}
   {\rm sf}\big(\mathfrak{\hat{h}}(t))_{t\in [0,1]}\;|\;\mu=0\big)\;&=\;{\rm sf}\big(\mathfrak{u}_0\mathfrak{\hat{h}}_1(t)\mathfrak{u}_0)_{t\in [0,1]}\;|\;\mu=0\big)\\
   &=\;{\rm sf}\big(\hat{H}(t)_{t\in [0,1]}\;|\;\mu=0\big)\\
   &=\;-\frac{1}{2\pi}\Delta\mathscr{P}
\end{split}
\end{equation*}
Thus, the Theorem \ref{teo: corner} implies that $\delta_0([\mathfrak{p_a}])=-\frac{1}{2\pi}\Delta\mathscr{P}$ and the claim follows.\\
Let us observe that this type of perturbations of the system depends implicitly on the geometry of the quarter plane. Namely, as we proved in Proposition \ref{Prop projections}, the projector $\mathfrak{u}_0$ lies in the algebra $\A$ as consequence of that $b_\llcorner-b_\ast\notin 2\pi\Z$ and the definition of the quarter-plane magnetic field $B.$
\end{example}


\providecommand{\href}[2]{#2}

\end{document}